%% file: logic-computation.tex
\documentclass[envcountsame]{llncs}

\newif\ifignore % when set to true, additional text appears containing
\ignorefalse
\newcommand{\auxproof}[1]{
\ifignore\mbox{}\newline
\textbf{PROOF:} \dotfill\newline
{\it #1}\mbox{}\newline
\textbf{ENDPROOF}\dotfill
\fi}

\usepackage{etex} % to get more space
\usepackage{amsmath}
\usepackage{amssymb}
\usepackage{amstext}
\usepackage{mathtools}
\usepackage{mathrsfs}
\usepackage{stmaryrd}
\usepackage{xspace}
\usepackage{cmll}
\usepackage{hyperref}
\usepackage{url}
\usepackage{fancybox}
\usepackage{wrapfig}

\usepackage{xypic}
\usepackage[all]{xy}
\xyoption{all}
\xyoption{2cell}
\UseTwocells
\xyoption{tips}
\SelectTips{cm}{}  % Tips (of arrows) are in accordance with Computer Modern
\usepackage{tikz}
\usetikzlibrary{circuits.ee.IEC}

\usepackage{array}
\newcolumntype{C}{>{\centering\arraybackslash}m{1.5cm}}

\usepackage{graphicx}

\newcommand{\id}{\mathrm{id}}
\newcommand{\strat}{\mathsf{Strat}}

\newcommand{\cato}{\Cat{O}}
\newcommand{\uly}{Y}
\newcommand{\parstrat}{\Psi}
\newcommand{\initsigma}{\mathsf{Init}\,\sigma}
\newcommand{\sigmaki}{\sigma/ki}
\newcommand{\setbr}[1]{\{{#1}\}}
\newcommand{\sumstar}{\sum^{\sharp}}
\newcommand{\lbim}[1]{{#1}^{\mathsf{Left}}}
\newcommand{\rbim}[1]{{#1}^{\mathsf{Right}}}
\newcommand{\eqdef}{\stackrel{\mbox{\rm {\tiny def}}}{=}}
\newcommand{\relto}{\to\hspace{-0.65em} \shortmid \hspace{0.5em}}

\newcommand{\pset}{\mathcal{P}}
\newcommand{\psetplus}{\mathcal{P}^{+}}
\newcommand{\arity}[1]{\mathsf{Ar}({#1})}
\newcommand{\betwixt }{\hspace{1em}}

\newcommand{\complatt}{\mathbf{CSL}}
\newcommand{\alcomplatt}{\mathbf{ACSL}}
\newcommand{\bigolift}{\bigoplus^{\bot}}
\input{prooftree}

\spnewtheorem{assumption}[theorem]{Assumption}{\bfseries}{\itshape}

\newcommand{\free}{\mathcal{F}}

\newcommand{\trlog}{\mathsf{log}}
\newcommand{\trem}{\mathsf{em}}
\newcommand{\trkl}{\mathsf{kl}}
\newcommand{\calogtb}{\ell_{\mathrm{log}}}
\newcommand{\calogbt}{\ell^{\mathrm{log}}}
\newcommand{\caem}{\ell_{\mathrm{em}}}
\newcommand{\cakl}{\ell_{\mathrm{kl}}}
\newcommand{\emlog}{\mathsf{e}}
\newcommand{\kllog}{\mathsf{k}}
\newcommand{\klem}{\mathsf{k}}

\newcommand{\QEDbox}{\square}
\newcommand{\QED}{\hspace*{\fill}$\QEDbox$}

\newcommand{\idmap}[1][]{\ensuremath{\mathrm{id}_{#1}}}
\newcommand{\Id}{\mathrm{Id}}
\newcommand{\after}{\mathrel{\circ}}
\newcommand{\evmap}{\mathrm{ev}}
\newcommand{\st}{\ensuremath{\mathrm{st}}}

\newcommand{\op}[1]{#1^{\mathrm{op}}}

\newcommand{\ket}[1]{\ensuremath{|{\kern.1em}#1{\kern.1em}\rangle}}
\newcommand{\bigket}[1]{\ensuremath{\big|{\kern.1em}#1{\kern.1em}\big\rangle}}

\newcommand{\powersetsymbol}{\mathcal{P}}

\newcommand{\finpowersetsymbol}{\powersetsymbol_{\textrm{f}}}

\newcommand{\distributionsymbol}{\mathcal{D}}

\newcommand{\findistributionsymbol}{\distributionsymbol_{\textrm{fin}}}

\newcommand{\Pow}{\powersetsymbol}
\newcommand{\Powfin}{\finpowersetsymbol}

\newcommand{\Dstfin}{\findistributionsymbol}

\newcommand{\UF}{\ensuremath{\mathcal{U}{\kern-.75ex}\mathcal{F}}}

\newcommand{\Cat}[1]{\ensuremath{\mathbf{#1}}\xspace}
\newcommand{\cat}[1]{\Cat{#1}}
\newcommand{\catc}{\Cat{C}}
\newcommand{\catd}{\Cat{D}}
\newcommand{\Kl}{\mathcal{K}{\kern-.4ex}\ell}
\newcommand{\EM}{\mathcal{E}{\kern-.4ex}\mathcal{M}}
\newcommand{\Alg}{\ensuremath{\mathrm{Alg}}}
\newcommand{\CoAlg}{\ensuremath{\mathrm{CoAlg}}}

\newcommand{\Sets}{\Cat{Sets}}

\newcommand{\Pred}{\ensuremath{\mathrm{Pred}}}

\newcommand{\Ef}{\ensuremath{\mathcal{E}{\kern-.5ex}f}}
\newcommand{\intd}{{\kern.2em}\mathrm{d}{\kern.03em}}

\newcommand{\OF}{\ensuremath{\mathcal{O}{\kern-.1em}\mathcal{F}}}
\newcommand{\Closed}{\ensuremath{\mathcal{C}{\kern-.05em}\ell}}

\newcommand{\congrightarrow}{\mathrel{\smash{\stackrel{
           \raisebox{.5ex}{$\scriptstyle\cong$}}{
           \raisebox{0ex}[0ex][0ex]{$\rightarrow$}}}}}

\makeatother
 \newdir{ >}{{}*!/-7.5pt/@{>}}
 \newdir{|>}{!/4.5pt/@{|}*:(1,-.2)@^{>}*:(1,+.2)@_{>}}
 \newdir{ |>}{{}*!/-3pt/@{|}*!/-7.5pt/:(1,-.2)@^{>}*!/-7.5pt/:(1,+.2)@_{>}}

\newsavebox\sbground
\savebox\sbground{\begin{tikzpicture}[circuit ee IEC,yscale=0.5,xscale=0.4]
                \draw (0,-2ex) to (0,0) node[ground,rotate=90,xshift=.65ex] {};
                \end{tikzpicture}}

\title{Steps and Traces\thanks{The research leading to these results has
    received funding from the European Research Council under the
    European Union's Seventh Framework Programme (FP7/2007-2013) / ERC
    grant agreement nr.~320571.
    This is a revised and extended version of a paper which appeared in the proceedings
    of CMCS 2018~\cite{JacobsLR18}.}}

\author{Jurriaan Rot\inst{1} and Bart Jacobs\inst{1} and Paul Levy\inst{2}}
\institute{Institute for Computing and Information Sciences,\\
Radboud Universiteit, Nijmegen, The Netherlands \\
\and University of Birmingham, UK \\
\email{jrot@cs.ru.nl} and \email{bart@cs.ru.nl} and \email{P.B.Levy@cs.bham.ac.uk}
}

\begin{document}
\maketitle
\begin{abstract}
In the theory of coalgebras, trace semantics can be defined in various
distinct ways, including through algebraic logics, the Kleisli category
of a monad or its Eilenberg-Moore category. This paper elaborates two
new unifying ideas: 1)~coalgebraic trace semantics is naturally
presented in terms of corecursive algebras, and 2)~all three approaches
arise as instances of the same abstract setting. Our perspective puts
the different approaches under a common roof, and allows to derive
conditions under which some of them coincide. 
\end{abstract}

\section{Introduction}\label{sec:intro}

Traces are used in the semantics of state-based systems as a way of
recording the consecutive behaviour of a state in terms of sequences
of observable (input and/or output) actions. Trace semantics leads to,
for instance, the notion of trace equivalence, which expresses that two
states cannot be distinguished by only looking at their iterated
in/output behaviour.

Trace semantics is a central topic of interest
in the coalgebra community---and not only there, of course. One of
the key features of the area of coalgebra is that states and their
coalgebras can be considered in different universes, formalised as
categories. The break-through insight is that trace semantics for a
system in universe $A$ can often be obtained by switching to a
different universe $B$. More explicitly, where the (ordinary)
behaviour of the system can be described via a final coalgebra in
universe $A$, the trace behaviour arises by finality in the different
universe $B$. Typically, the alternative universe $B$ is a category of
algebraic logics, the Kleisli category, or the category of
Eilenberg-Moore algebras, of a monad on universe $A$.

This paper elaborates two new unifying ideas.
\begin{enumerate}
\item We observe that the trace map from the state space of a
  coalgebra to a carrier of traces is in all three situations the
  unique `coalgebra-to-algebra' map to a \emph{corecursive
    algebra}~\cite{CaprettaUV09} of traces. This differs from earlier
  work which tries to describe traces as final coalgebras. For us it
  is quite natural to view languages as algebras, certainly when they
  consist of finite words/traces.

\item Next, these corecursive algebras, used as spaces of traces, all
  arise via a uniform construction, in a setting given by an
  adjunction together with a special natural transformation that we
  call a `step'. We heavily rely on a basic result saying that in this
  situation, the (lifting of the) right adjoint preserves corecursive
  algebras, sending them from one universe to another. This is a known
  result~\cite{CaprettaUustaluVene:recursive}, but its fundamental
  role in trace semantics has not been recognized before. For an
  arbitrary coalgebra there is then a unique map to the transferred
  corecursive algebra; this is the trace map that we are after.
\end{enumerate}

\noindent The main contribution of this paper is the unifying
step-based approach to coalgebraic trace semantics: it is shown that
three existing flavours of trace semantics---logical,
Eilenberg-Moore, Kleisli---are all instances of our
approach. Moreover, comparison results are given relating 
theses. 
We focus only on finite
trace semantics, and also exclude at this stage the `iteration' based
approaches, e.g.,~in~\cite{MiliusPS15,DorschMS19,KurzMPS15,Cirstea16}.

\paragraph{Outline.} The paper is organised as follows. It starts in
Section~\ref{sec:prelim} with the abstract step-and-adjunction setting,
and the relevant definitions and results for corecursive algebras. 
In the next three sections, it is explained how this 
setting gives rise to trace semantics, by presenting
the above-mentioned three approaches to coalgebraic trace semantics
in terms of steps and adjunctions: Eilenberg-Moore (Section~\ref{sec:em}),
logical (Section~\ref{sec:logic}) and Kleisli (Section~\ref{sec:kl}).
In each case, the relevant corecursive algebra is described. 
These sections are illustrated with several examples. 
In Section~\ref{sec:partial} we study
partial traces for coalgebras with input and output~\cite{BowlerLevyPlotkin:infinstrat}, as another instance of the step-and-adjunction setting; but it is helpful to express that setting in the language of bimodules, which we do in Appendix~\ref{sec:bimod}.

The next section establishes a connection between the Eilenberg-Moore and the logical approach,
between the Kleisli and logical approach and between the Kleisli and Eilenberg-Moore approach (Section~\ref{sec:comp}).
In Section~\ref{sec:cia} we show that our construction yields algebras that are not merely corecursive but completely iterative, a stronger property that provides more general trace semantics. Finally, in Section~\ref{sec:fw} we
provide some directions for future work.

\paragraph{Notation.} \label{sec:prelim}

In the context of an adjunction $F \dashv G$, we shall use overline
notation $\overline{(-)}$ for adjoint transposition. The unit and
counit of an adjunction are, as usual, written as $\eta$ and
$\varepsilon$.

For an endofunctor $H$, we write $\Alg(H)$ for its algebra category
and $\CoAlg(H)$ for its coalgebra category.  For a monad
$(T,\eta,\mu)$ on $\catc$, we write $\EM(T)$ for the Eilenberg-Moore
category and $\Kl(T)$ for the Kleisli category.

We recall that any functor $S \colon \Sets \rightarrow \Sets$ has a
unique strength $\st$.  We write $\st \colon S(X^A) \rightarrow
S(X)^A$ for $\st(t)(a) = S(\evmap_a)(t)$, where $\evmap_a = \lambda
f. f(a) \colon X^A \rightarrow X$.

%  for any sets $X$ and $A$ the strength
% map $\st \colon S(X^A) \rightarrow S(X)^A$, natural in $X$ and $A$,
% defined by $\st(t)(a) = S(\evmap_a)(t)$, where $\evmap_a = \lambda
% f. f(a) \colon X^A \rightarrow X$.

\section{Coalgebraic semantics from a step} \label{sect:stepsem}

This section is about the construction of corecursive algebras and
their use for semantics.  The notion of corecursive algebra, studied
in~\cite{Eppendahl:coalgtoalg,CaprettaUV09} as the dual of Taylor's
notion of recursive coalgebra~\cite{GirardLafontTaylor:book}, is
defined as follows.

\begin{definition} Let $H$ be an endofunctor on a category $\catc$.
\begin{enumerate}
\item A \emph{coalgebra-to-algebra morphism} from a coalgebra $c\colon
  X \rightarrow H(X)$ to an algebra $a\colon H(\Theta)\rightarrow \Theta$ is a
  map $f\colon X \rightarrow \Theta$ such that the diagram
\[ \xymatrix@R-0.5pc{
 X \ar[r]^-{f} \ar[d]_{c} & \Theta  \\
 H(X) \ar[r]_-{H(f)} & H(\Theta) \ar[u]_{a}
} \]

\noindent commutes.  Equivalently: such a morphism is a fixpoint for the
endofunction on the homset $\catc(X,\Theta)$ sending $f$ to the composite
\[ \xymatrix{
 X \ar[r]^-{c} & H(X) \ar[r]^-{H(f)} & H(\Theta) \ar[r]^-{a} & \Theta
} \]

\item An algebra $a \colon H(A) \rightarrow A$ is \emph{corecursive}
  when for every coalgebra $c \colon X \rightarrow H(X)$ there is a
  unique coalgebra-to-algebra morphism $(X,c) \to (\Theta,a)$.
  \end{enumerate}
\end{definition}

Here is some intuition.

\begin{itemize}
\item As explained in~\cite{HinzeWuGibbons:conjhylo}, the
  specification of a coalgebra-to-algebra morphism $f$ is a
  ``divide-and-conquer'' algorithm.  It says: to operate on an
  argument, first decompose it via the coalgebra $c$, then operate on
  each component via $H(f)$, then combine the results via the algebra
  $a$.

\item For each final $H$-coalgebra $\zeta \colon \Theta
  \congrightarrow H(\Theta)$, the inverse $\zeta^{-1} \colon H(\Theta)
  \rightarrow \Theta$ is a corecursive algebra. For most functors of
  interest, this final coalgebra gives semantics up to bisimilarity,
  which is finer than trace equivalence.  So trace semantics requires
  a different corecursive algebra.
\end{itemize}

In all our examples, we use the same procedure for obtaining a
corecursive algebra. It makes frequent use of the following so-called
\emph{mate
  correspondence}~\cite{KellyStreet:twocat,leinster2004higher}; also
  see, e.g.,~\cite{Klin07,JacobsS10,KlinR16,Levy15} for special cases.

\begin{theorem} 
\label{thm:matesgeneral}
Given adjunctions and functors
  \begin{displaymath}
    \xymatrix{
      \cat{C}\ar@/^2ex/[rr]^-{F}  \ar[d]_{H} & \bot & \cat{D} \ar@/^2ex/[ll]^-G \ar[d]^{L} \\
       \cat{C'}\ar@/^2ex/[rr]^-{F'}  & \bot & \cat{D'} \ar@/^2ex/[ll]^-{G'}
      }
    \end{displaymath}

\noindent there are bijective correspondences between natural transformations:
    \begin{displaymath}
      \xymatrix@R-0.5pc{
        \cat{C} \ar[d]_{H} \ar@{}[dr]|{\stackrel{\rho_1}{\Longrightarrow}} &  \cat{D} \ar[d]^{L} \ar[l]_{G} & \cat{C}\ar[r]^{F} \ar[d]_{H} \ar@{}[dr]|{\stackrel{\rho_2}{\Longrightarrow}} &  \cat{D} \ar[d]^{L} & \cat{C} \ar[d]_{H}  \ar@{}[dr]|{\stackrel{\rho_3}{\Longrightarrow}} &  \cat{D} \ar[d]^{L}  \ar[l]_{G} & \cat{C} \ar[d]_{H} \ar[r]^{F} \ar@{}[dr]|{\stackrel{\rho_4}{\Longrightarrow}} &  \cat{D} \ar[d]^{L} \\
        \cat{C'} & \cat{D'} \ar[l]^{G'} &  \cat{C'} \ar[r]_{F'} & \cat{D'} &  \cat{C'}\ar[r]_{F'} & \cat{D'} &  \cat{C'} & \cat{D'} \ar[l]^{G'}
        }
    \end{displaymath}
%  $\rho_1 \colon HG
% \Rightarrow G'L$, $\rho_2 \colon F'H \Rightarrow LF$, $\rho_3 \colon F'HG
% \Rightarrow L$ and $\rho_4 \colon H \Rightarrow G'LF$.

\noindent Here $\rho_1$ and $\rho_3$ correspond by adjoint
transposition, and similarly for $\rho_2$ and $\rho_4$. Further,
$\rho_1$ and $\rho_2$ are obtained from each other by:
\begin{displaymath}
\begin{array}[b]{rcl}
\rho_1
& = &
\left( \xymatrix{HG \ar@{=>}[r]^-{\eta' HG} & G'F'HG \ar@{=>}[r]^-{G'\rho_2 G} & G'LFG \ar@{=>}[r]^-{G'L\varepsilon} & G'L}\right) 
\\
\rho_2
& = &
\left( \xymatrix{F'H \ar@{=>}[r]^-{F'H\eta} & F'HGF \ar@{=>}[r]^-{F'\rho_1 F} & F'G'LF \ar@{=>}[r]^-{\varepsilon' LF} & LF} \right) \,.
\end{array}
\end{displaymath}
\end{theorem}
It is common to refer to $\rho_1$ and $\rho_2$ as \emph{mates}; the
other two maps are their adjoint transposes, as we have seen.  In diagrams we omit the
subscript $i$ in $\rho_i$ and let the type determine which version of
$\rho$ is meant.  Further, in the remainder of this paper we usually drop the subscript
of components of natural transformations.

Our basic setting
consists of an adjunction, two endofunctors, and a natural
transformation:
\begin{equation}
\label{eq:adj-setting}
  \xymatrix{
\cat{C}\ar@/^2ex/[rr]^-{F} \save !L(.5) \ar@(dl,ul)^{H} \restore & \bot & \cat{D} \ar@/^2ex/[ll]^-G \save !R(.5) \ar@(ur,dr)^{L} \restore 
& \mbox{\qquad with} &
HG \ar@{=>}[r]^-{\rho} & GL
}
\end{equation}

\noindent The natural transformation $\rho \colon HG \Rightarrow GL$
will be called a \emph{step}.  Here $H$ is the \emph{behaviour
  functor}: we study $H$-coalgebras and give semantics for them in a
corecursive $H$-algebra.  This arrangement is well-known in the area
of coalgebraic modal
logic~\cite{BonsangueK05,PavlovicMW06,Klin07,ChenJung:catframework,Levy15},
but we shall see that its application is wider.
The following result shows different equivalent presentations of a step; for the proof,
see Appendix~\ref{app:details-steps}.

\begin{theorem}	
\label{thm:mates}
In the situation~\eqref{eq:adj-setting}, there are bijective
correspondences between natural transformations $\rho_1 \colon HG
\Rightarrow GL$, $\rho_2 \colon FH \Rightarrow LF$, $\rho_3 \colon FHG
\Rightarrow L$ and $\rho_4 \colon H \Rightarrow GLF$, as in Theorem~\ref{thm:matesgeneral}.
	
Moreover, if $H$ and $L$ happen to be monads, then $\rho_1$ is an
$\EM$-law (map $HG \Rightarrow GL$ compatible with the monad
structures) iff $\rho_2$ is a $\Kl$-law (map $FH \Rightarrow LF$
compatible with the monad structures) iff $\rho_4$ is a monad map; and
two further equivalent characterisations are respectively a lifting of
$G$ or an extension of $F$:
\[ \xymatrix@R-0.5pc{
	\EM(H) \ar[d]
		& \EM(L) \ar[d] \ar[l]_{\overline{G}} \\
	\Cat{C} 
		& \Cat{D} \ar[l]_G
	}
	\qquad\qquad
	\xymatrix@R-0.5pc{
	\Kl(H) \ar[r]^{\overline{F}}
		& \Kl(L) \\
	\Cat{C} \ar[r]^{F} \ar[u]
		& \Cat{D}  \ar[u]
	} \eqno{\raisebox{-2.2em}{$\QEDbox$}} \]
\end{theorem}

% \begin{proof}
% We only mention the bijective correspondences: $\rho_1$ and $\rho_3$
% correspond by adjoint transposition, and similarly for $\rho_2$ and
% $\rho_4$. Further, $\rho_1$ and $\rho_2$ are obtained from each other
% by:
% $$\begin{array}[b]{rcl}
% \rho_1
% & = &
% \left( \xymatrix{HG \ar@{=>}[r]^-{\eta HG} & GFHG \ar@{=>}[r]^-{G\rho_2 G} & GLFG \ar@{=>}[r]^-{GL\varepsilon} & GL}\right) 
% \\
% \rho_2
% & = &
% \left( \xymatrix{FH \ar@{=>}[r]^-{FH\eta} & FHGF \ar@{=>}[r]^-{F\rho_1 F} & FGLF \ar@{=>}[r]^-{\varepsilon LF} & LF} \right) \,.
% \end{array} \eqno{\QEDbox}$$
% \end{proof}

Steps give rise to liftings to categories of algebras and coalgebras, as follows. 

\begin{definition}\label{def:step-liftings}
In the setting~\eqref{eq:adj-setting}, the step natural transformation
$\rho$ gives rise to both:
  \begin{itemize}
  \item a lifting $G_\rho$ of the right adjoint $G$, called the
    \emph{step-induced algebra lifting}:
$$
\vcenter{ \xymatrix@R-0.5pc@C-0.5pc{
    \Alg(H)\ar[d]  & \Alg(L)\ar[d] \ar[l]_-{G_{\rho}}\\
    \cat{C} & \cat{D}\ar[l]_-{G} } } \qquad\qquad
\begin{array}{l}
  G_{\rho} \Big(L(X) \xrightarrow{a} X\Big) \;\coloneqq\; \\
\quad  \Big(HG(X) \xrightarrow{\rho} 
  GL(X)\xrightarrow{G(a)}  G(X)\Big).
\end{array}
$$
\item dually, a lifting $F^{\rho}$ of the left adjoint $F$, called the
  \emph{step-induced coalgebra lifting}:
$$
\vcenter{ \xymatrix@R-0.5pc@C-0.5pc{
    \CoAlg(H)\ar[d]  \ar[r]^-{F^{\rho}} & \CoAlg(L)\ar[d]\\
    \cat{C}\ar[r]^-{F} & \cat{D} } } \qquad\qquad
\begin{array}{l}
  F^{\rho} \Big(X \xrightarrow{c} H(X)\Big) \;\coloneqq\; \\
\quad \Big(F(X) \xrightarrow{F(c)} FH(X) \xrightarrow{\rho} LF(X) \Big).
\end{array}
$$
\end{itemize}
\end{definition}

Our approach relies on the following basic result.

\begin{proposition}[\cite{CaprettaUustaluVene:recursive}]
\label{prop:pres-corec}
In the setting~\eqref{eq:adj-setting}, for a corecursive $L$-algebra
$a \colon L(\Theta) \to \Theta$, the transferred $H$-algebra $G_{\rho}
(\Theta,a) \colon HG(\Theta) \to G(\Theta)$ is also corecursive.
\end{proposition}

\begin{proof}
Let $c\colon X \rightarrow H(X)$ be an $H$-coalgebra. Then $F^{\rho}(X,c)$ is
an $L$-coalgebra, which gives rise to a unique coalgebra-to-algebra
map $f\colon F(X) \rightarrow \Theta$ with $a \after L(f) \after \rho
\after F(c) = f$. The adjoint-transpose $g\colon X \rightarrow
G(\Theta)$ of $f$ is then the unique coalgebra-to-algebra map from
$(X,c)$ to $G_{\rho}(\Theta,a)$. \QED
\end{proof}

Thus, by analogy with the familiar statement that \emph{``right
  adjoints preserves limits''}, we have \emph{``step-induced algebra
  liftings of right adjoints preserve corecursiveness''}.  Now we give
the complete construction for semantics of a coalgebra.

\begin{theorem} \label{thm:finalcor}
Suppose that $L$ has a final coalgebra $\smash{\zeta\colon \Psi
  \congrightarrow L(\Psi)}$. Then for every $H$-coalgebra $(X,c)$
there is a unique coalgebra-to-algebra map $c^\dagger$ as on the left
below:
	$$
	\xymatrix@R-.5pc{
		X \ar@{-->}[r]^-{c^\dagger} \ar[d]_{c} 
			& G(\Psi)  \\
		H(X) \ar@{-->}[r]^-{H(c^\dagger)}
			& HG(\Psi) \ar[u]_{G_\rho(\Psi,\zeta^{-1})}
	}
	\qquad\qquad
	\xymatrix@R-.5pc{
		F(X) \ar@{-->}[r]^-{\overline{c^\dagger}} \ar[d]_{F^\rho(X,c)} 
			& \Psi  \\
		LF(X) \ar@{-->}[r]^-{L(\overline{c^\dagger})}
			& L(\Psi) \ar[u]_{\zeta^{-1}}
	}
	$$

\noindent The map $c^\dagger$ on the left can alternatively be
characterized via its adjoint transpose $\overline{c^{\dagger}}$ on
the right, which is the unique coalgebra-to-algebra morphism.
The latter can also be seen as the unique map to the final
coalgebra $\Psi \congrightarrow L(\Psi)$. \QED
\end{theorem}

Note that Theorem~\ref{thm:finalcor} generalises final coalgebra
semantics: taking in (\ref{eq:adj-setting}) $F = G = \Id_{\cat{C}}$
and $H = L$, the map $c^\dagger$ in the above theorem is the unique
homomorphism to the final coalgebra.  In the remainder of this paper
we focus on instances where $c^\dagger$ captures traces, and we
therefore refer to $c^\dagger$ as the \emph{trace semantics} map.

Given steps $\rho \colon HG \Rightarrow GL$ and $\theta \colon KG
\Rightarrow GM$, we can form a new step by composition, written as
$\theta \circledcirc \rho$ in:
\begin{equation}
\label{eqn:composition-steps}
\begin{array}{rcl}
\theta \circledcirc \rho
& \coloneqq &
\xymatrix{\Big(KHG\ar@{=>}[r]^-{K\rho} &
   KGL\ar@{=>}[r]^-{\theta L} & GML\Big) }
\end{array}
\end{equation}

We conclude with a lemma that relates the mate construction to composition of steps. 
See Appendix~\ref{app:details-steps} for a proof. 

\begin{lemma}
\label{lm:mates-composition}
Let $\rho \colon HG \Rightarrow GL$, $\theta \colon KG \Rightarrow GM$
be steps. Then $\big(\theta \circledcirc \rho\big)_2 = M
\rho_2 \after \theta_2 H$.
\end{lemma}

\section{Traces via Eilenberg-Moore}
\label{sec:em}

We recall the approach to trace semantics  
developed in~\cite{JacobsSS15,SBBR13,BonsangueMS13}, putting it in the framework of the previous section. 
The approach deals with coalgebras for the composite functor $BT$,
where $T$ is a monad that captures the `branching' aspect. 
The following assumptions are required. 
\begin{assumption}[Traces via Eilenberg-Moore]\label{ass:em} In this section, we assume:
\begin{enumerate}
\item An endofunctor $B \colon \cat{C} \rightarrow \cat{C}$ with a
  final coalgebra $\zeta \colon \Theta \congrightarrow B(\Theta)$.

\item A monad $(T,\eta,\mu)$, with the standard adjunction $\free
  \dashv U$ between categories $\cat{C} \leftrightarrows \EM(T)$,
  where $U$ is `forget' and $\free$ is for `free algebras'.

\item A lifting $\overline{B}$ of $B$, as in:
	\begin{equation}\label{eq:lifting-em}
\vcenter{
	\xymatrix@R-0.5pc{
		\EM(T) \ar[r]^{\overline{B}} \ar[d]_U
			& \EM(T)  \ar[d]^U \\
		\Cat{C} \ar[r]_{B}
			& \Cat{C}
	}
}
\end{equation}

\noindent or, equivalently, an $\EM$-law $\kappa \colon TB \Rightarrow
BT$.
\end{enumerate}
\end{assumption}

\begin{example}
\label{ex:first-em}
To briefly illustrate these ingredients, we consider non-deterministic
automata.  These are $BT$-coalgebras with $B\colon \Sets \rightarrow
\Sets$, $B(X) = 2 \times X^A$ where $2 = \{\bot,\top\}$ and $T$ the
finite powerset monad.  The functor $B$ has a final coalgebra carried
by the set $2^{A^*}$ of languages.  Further, $\EM(T)$ is the category
of join semi-lattices (JSLs). The lifting is defined by products in
$\EM(T)$, using the JSL on $2$ given by the usual ordering $\bot \leq
\top$.  By the end of this section, we revisit this example and obtain
the usual language semantics.
\end{example}

The above three assumptions give rise to the following instance of our
general setting~\eqref{eq:adj-setting}:
\begin{equation}\label{eq:em-standard}
\xymatrix@C-0.5pc{
\cat{C}\ar@/^2ex/[rr]^-{\free} \save !L(.5) \ar@(dl,ul)^{BT} \restore & \bot & \EM(T) \ar@/^2ex/[ll]^-U \save !R(.5) \ar@(ur,dr)^{\overline{B}} \restore 
& \quad\mbox{with}\quad
{\begin{array}{l}
\rho\colon BTU \Longrightarrow U\overline{B} \; \mbox{ where } 
\\
\rho_{(X,a)} = \big(BTX \xrightarrow{Ba} BX\big)
\end{array}}
}
\end{equation}

\noindent Actually---and equivalently, by Theorem~\ref{thm:mates}---the 
step $\rho$ is most easily given in terms of $\rho_4 \colon BT
\Rightarrow U\overline{B}\free$: since $\overline{B}$ lifts $B$, we
have $U\overline{B}\free = BU\free = BT$, so that $\rho_4$ is then
defined simply as the identity.

The following result is well-known, and is (in a small variation) due
to~\cite{TuriP97}.

\begin{lemma}
\label{lm:final-coalg-em}
There is a unique algebra structure $a \colon T(\Theta) \rightarrow
\Theta$ making $((\Theta,a),\zeta)$ a
$\overline{B}$-coalgebra. Moreover, this coalgebra is final in
$\CoAlg(\overline{B})$. \QED
\end{lemma}

\begin{proof}
We recall that the map $a\colon T(\Theta) \rightarrow \Theta$ is obtained
by finality in:
\begin{equation}
\label{eq:final-coalg-em}
\vcenter{\xymatrix@R-0.5pc{
T(\Theta)\ar[d]_{\kappa \after T(\zeta)}\ar@{-->}[r]^-{a} & \Theta\ar[d]^{\zeta}_{\cong}
\\
BT(\Theta)\ar@{-->}[r]_-{B(a)} & B(\Theta)
}}
\end{equation}

\noindent This gives an Eilenberg-Moore algebra $(\Theta,a)$, with a
$\overline{B}$-coalgebra $\zeta \colon (\Theta,a) \rightarrow
\overline{B}(\Theta, a)$ which is final. \QED
\end{proof}

We apply the step-induced algebra lifting $G_\rho \colon
\Alg(\overline{B}) \rightarrow \Alg(BT)$ to the inverse of this final
$\overline{B}$-coalgebra, obtaining a $BT$-algebra:
\begin{equation}
\label{eqn:caem}
\begin{array}{rcccl}
\big(BT(\Theta) \xrightarrow{\caem} \Theta\big) 
& \coloneqq &
G_\rho((\Theta,a),\zeta^{-1})
& = &
\big(BT(\Theta) \xrightarrow{B(a)} B(\Theta) \xrightarrow{\zeta^{-1}}\Theta\big).
\end{array}
\end{equation}

\noindent By Theorem~\ref{thm:finalcor}, this $BT$-algebra $\caem$ is
corecursive, giving us trace semantics of $BT$-coalgebras. More
explicitly, given a coalgebra $c \colon X \rightarrow BT(X)$, the
trace semantics is the unique map, written as $\trem_c$, making the
following square commute.
\begin{equation}
\label{eq:traces-em}
\vcenter{\xymatrix@R-0.5pc{
	X \ar@{-->}[r]^{\trem_c} \ar[d]_c
		& \Theta \\
	BT(X) \ar@{-->}[r]_{BT(\trem_c)}
		& BT(\Theta) \ar[u]_{\caem}
}}
\end{equation}

\noindent The unique map $\trem_c$ in~\eqref{eq:traces-em} appears in
the literature as a `coiteration up-to' or `unique solution'
theorem~\cite{Bartels03}.  Examples follow later in this section
(Theorem~\ref{thm:em-automata}, Example~\ref{ex:em-automata}).

In~\cite{JacobsSS15,SBBR13}, the above trace semantics of
$BT$-coalgebras arises through `determinisation', which we explain
next. Given a coalgebra $c \colon X \rightarrow BT(X)$, one takes its
adjoint transpose:
$$
\frac{c \colon X \longrightarrow BT(X) = BU\free(X) = U\overline{B}\free(X)}
{\overline{c} \colon \free(X) \longrightarrow \overline{B}\free(X)}
$$

It follows from Theorem~\ref{thm:mates} and our definition of $\rho$
that this transpose coincides with the application of the step-induced
coalgebra lifting $\free^\rho \colon \CoAlg(BT) \rightarrow
\CoAlg(\overline{B})$ from the previous section, i.e.,
$\free^\rho(X,c) = (\free(X),\overline{c})$. The functor $\free^\rho$
thus plays the role of determinisation, see~\cite{JacobsSS15}.  By
Theorem~\ref{thm:finalcor}, the trace semantics $\trem_c$ can
equivalently be characterised in terms of $\free^\rho$, as the unique
map $\overline{\trem_c}$ making the diagram below commute.  
\begin{equation}
\label{eq:det}
\vcenter{
\xymatrix@R-0.5pc{
	T(X) \ar@{-->}[r]^{\overline{\trem_c}} \ar[d]_{\overline{c}}
		& \Theta \\
	BT(X) \ar@{-->}[r]_{B(\overline{\trem_c})}
		& B(\Theta) \ar[u]_{\zeta^{-1}}
}
}
\end{equation}

\noindent This is how the trace semantics via Eilenberg-Moore is
presented in~\cite{JacobsSS15,SBBR13}: as the transpose $\trem_c =
\overline{\trem_c} \after \eta_{X} \colon X \rightarrow \Theta$.

We conclude this section by recalling a canonical construction of a
distributive law~\cite{Jacobs06} for a class of `automata-like'
examples that we will spell out in Example~\ref{ex:em-automata}.

\begin{theorem}\label{thm:em-automata}
Let $\Omega$ be a set, $T$ a monad on $\Sets$ and $t \colon T(\Omega)
\rightarrow \Omega$ an $\EM$-algebra.  Let $B \colon \Sets \rightarrow
\Sets$, $B(X) = \Omega \times X^A$, and $\kappa \colon TB \Rightarrow
BT$ given by
	$$
	\kappa_X :=
	\Big(
	\xymatrix@C=1.5cm{
		T(\Omega \times X^A) \ar[r]^-{\langle T(\pi_1),T(\pi_2)\rangle}
			& T(\Omega) \times T(X^A) \ar[r]^-{t \times \st}
			& \Omega \times T(X)^A
	}
	\Big) \,.
	$$
	Then $\kappa$ is an $\EM$-law.
	Moreover, the algebra structure
	on the carrier of the final coalgebra $(\Omega^{A^*}, \zeta)$ mentioned in 
	the statement of
	Lemma~\ref{lm:final-coalg-em} is given by
	$
	\xymatrix@C-0.5pc{
		T(\Omega^{A^*}) \ar[r]^-{\st}
			& T(\Omega)^{A^*} \ar[r]^-{t^{A^*}}
			& \Omega^{A^*}
	}
	$. Hence, this algebra is the carrier of a final $\overline{B}$-coalgebra.
	\QED
\end{theorem}
%More specifically, the second part of the theorem characterises the
%algebra induced on the final $B$-coalgebra in~\eqref{eq:final-coalg-em}, and hence 
%by Lemma~\ref{lm:final-coalg-em} it follows that this algebra is the carrier of the final $\overline{B}$-coalgebra. 

\begin{example}
\label{ex:em-automata}
By Theorem~\ref{thm:em-automata}, we obtain an explicit description
of the trace semantics arising from the corecursive algebra~\eqref{eq:traces-em}:
for any $\langle o, f \rangle \colon X \rightarrow \Omega \times T(X)^A$,
the trace semantics is the unique map $\trem$ in:
$$
\xymatrix@R-0.5pc{
	X \ar@{-->}[rrr]^-{\trem} \ar[d]_-{\langle o,f \rangle}
		& & & \Omega^{A^*} \\
	BT(X) \ar@{-->}[r]_-{BT(\trem)}
		& BT(\Omega^{A^*}) \ar[r]_-{B(\st)}
		& B(T(\Omega)^{A^*}) \ar[r]_-{B(t^{A^*})}
		& B(\Omega^{A^*}) \ar[u]_-{\zeta^{-1}}
}
$$ We instantiate the trace semantics $\trem$ for various choices of
$\Omega$, $T$ and $t$. Given a coalgebra $\langle o, f \rangle \colon
X \rightarrow \Omega \times T(X)^A$, we have $\trem(x)(\varepsilon) =
o(x)$ independently of these choices. The table below lists the
inductive case $\trem(x)(aw)$ respectively for non-deterministic
automata (NDA) where branching is interpreted as usual
(NDA-$\exists$), NDA where branching is interpreted conjunctively
(NDA-$\forall$) and (reactive) probabilistic automata (PA).  Here
$\Powfin$ is the finite powerset monad, and $\Dstfin$ the finitely
supported distribution (or subdistribution) monad.
$$
\begin{array}{c|c|c|c|c}
		& T & \Omega & t \colon T(\Omega) \rightarrow \Omega & \trem(x)(aw) \\
		\hline 
		\text{NDA-}\exists
		& \Powfin
		& 2 = \{\bot,\top\}
		& S \mapsto \bigvee S
		& \bigvee_{y \in f(x)(a)} \trem(y)(w) \\
		\text{NDA-}\forall
		& \Powfin
		& 2 = \{\bot,\top\}
		& S \mapsto \bigwedge S
		& \bigwedge_{y \in f(x)(a)} \trem(y)(w) \\
		\text{PA}
		& \Dstfin
		& [0,1]
		& \varphi \mapsto \sum_{p \in [0,1]}p \cdot \varphi(p) 
		& \sum_{y \in X} \trem(y)(w) \cdot f(x)(a)(y) 
\end{array}
$$
For other examples, and a concrete presentation of the associated
determinisation constructions, see~\cite{JacobsSS15,SBBR13}.
\end{example}

\subsection{Eilenberg-Moore trace semantics for $TA$-coalgebras}\label{sec:gen-em}

We now extend the above treatment of trace semantics of $BT$-coalgebras via Eilenberg-Moore categories,
to cover coalgebras for a composite functor $TA$ as well, where $A$ is another endofunctor 
on the base category $\cat{C}$. This integrates the \emph{extension semantics} of~\cite{JacobsSS15} in the present setting;
the latter covers examples such as non-deterministic automata (as in Example~\ref{ex:nda-generative}) and probabilistic systems
in generative form.
The approach to trace semantics of $TA$-coalgebras in this section extends Assumption~\ref{ass:em},
making use of a lifting $\overline{B}$ of a functor $B$ to obtain traces as a suitable final coalgebra. 
Note that $A$ itself is not lifted, but is connected to $B$ via a step $\rho$ as stipulated in the global
assumptions of this subsection, described next.

\begin{assumption}\label{ass:emta} In addition to Assumption~\ref{ass:em},
we assume a functor $A \colon \cat{C} \rightarrow \cat{C}$, and a step $\rho$
in:
$$
\xymatrix@C-0.5pc{
\cat{C}\ar@/^2ex/[rr]^-{\free} \save !L(.5) \ar@(dl,ul)^{A} \restore & \bot & \EM(T) \ar@/^2ex/[ll]^-U \save !R(.5) \ar@(ur,dr)^{\overline{B}} \restore 
& \qquad\mbox{with}\qquad
{\begin{array}{l}
\rho \colon AU \Longrightarrow U \overline{B} \; 
\end{array}}
}
$$ 
\end{assumption}

\noindent 
The counit $\varepsilon \colon \free U \rightarrow U$ is given by
$\varepsilon_{(X,a)} = a$; notice that $\free U (X,a) = (T(X), \mu_X)$. 
Applying the forgetful functor to $\varepsilon$
gives another step $U \varepsilon \colon TU \Rightarrow U$,
where the `$L$' (from~\eqref{eq:adj-setting}) in the codomain is the
identity functor.  We can compose the steps $U\varepsilon$ and $\rho$
in two ways. First, we get a step for $AT$ by composing as follows:
\[ \begin{array}{rcl}
\rho \circledcirc U\varepsilon
& = &
\xymatrix{\Big(ATU\ar@{=>}[r]^-{AU\varepsilon} & 
   AU\ar@{=>}[r]^-{\rho} & U \overline{B} \Big)}
\end{array} \]

\noindent If $A=B$, then taking $\rho$ to be the identity is precisely
the step defined in~\eqref{eq:em-standard}.

We turn to the other composition of $U\varepsilon$ and $\rho$, which
gives a step for $TA$:
\[ \begin{array}{rcl}
U\varepsilon \circledcirc \rho
& = &
\xymatrix{\Big(TAU\ar@{=>}[r]^-{T\rho} &
 TU  \overline{B}  \ar@{=>}[r]^-{U\varepsilon  \overline{B} } & U  \overline{B}  \Big) }
\end{array} \]

As we will see in Proposition~\ref{prop:extension-step}, 
steps $\rho$ as in Assumption~\ref{ass:emta} correspond to natural transformations of the form
$\mathfrak{e}
\colon TA \Rightarrow BT$ making the following diagram commute:
\begin{equation}
\label{eqn:extensiondiagsingle}
\vcenter{\xymatrix@R-0.5pc{
T^{2}A\ar[rr]^-{\mu}\ar[d]_{T(\mathfrak{e})} & & TA\ar[d]^{\mathfrak{e}}
\\
TBT\ar[r]^-{\kappa} & BT^{2}\ar[r]^-{B(\mu)} & BT
}}
\end{equation}
In~\cite{JacobsSS15}, a natural transformation $\mathfrak{e}$ making this diagram commute is called an 
\emph{extension law} if it additionally satisfies a coherence axiom with a $\Kl$-law. 
We will only see this later, in our
comparison between different approaches for assigning trace semantics,
see Section~\ref{sec:comp-kleisli-em}.
The last line of the correspondence below involves natural transformations of the form
$A \Rightarrow BT$, which are called \emph{generic observers} in~\cite{Goncharov13}.

\begin{proposition}
\label{prop:extension-step}
There is a one-to-one correspondence between:
$$\begin{prooftree}
\begin{prooftree}
\text{steps }\; \rho \colon AU \Longrightarrow U\overline{B}
\Justifies
\mathfrak{e} \colon TA \Rightarrow BT 
  \;\text{ satisfying~\eqref{eqn:extensiondiagsingle}}
\end{prooftree}
\Justifies
A \Longrightarrow BT
\end{prooftree}$$
\end{proposition}

\begin{proof}
By Theorem~\ref{thm:mates} the natural transformation $\rho = \rho_{1}
\colon AU \Rightarrow U\overline{B}$ corresponds to $\mathfrak{e} =
\rho_{2} \colon \free A \Rightarrow \overline{B}\free$; the latter is
a natural transformation $TA \Rightarrow BT$ whose components are maps
of algebras $\mu_{X} \rightarrow \overline{B}(\mu_{x})$, as expressed
by Diagram~\eqref{eqn:extensiondiagsingle}. This covers the first
correspondence in the proposition.

By Theorem~\ref{thm:mates}, a natural transformation $\mathfrak{e} \colon
TA \Rightarrow BT$ further corresponds to a natural transformation $A
\Rightarrow U\overline{B} \free = BU \free = BT$. The latter is simply
a natural transformation on the base category $\Cat{C}$, which means
no further coherence axioms like~\eqref{eqn:extensiondiagsingle} need
to be checked. \QED

\auxproof{ In more detail, given $\mathfrak{e} \colon TA \Rightarrow BT$
  and an Eilenberg-Moore algebra $a\colon T(X) \rightarrow X$ we get a
  map:
\[ \xymatrix{
A(X)\ar[r]^-{\eta} & TA(X)\ar[r]^-{\mathfrak{e}} & BT(X)\ar[r]^-{B(a)} & B(X)
} \]

\noindent These maps form a natural transformation $AU \Rightarrow
U\overline{B}$.

Let's write $\overline{\mathfrak{e}}_{(X,a)}$ for the above map, and let
$f\colon (X,a) \rightarrow (Y,b)$ be a map of algebras. Naturality of
$\overline{\mathfrak{e}}$ is shown in:
\[ \begin{array}{rcl}
U\overline{B}(f) \after \overline{\mathfrak{e}}_{(X,a)}
& = &
B(f) \after B(a) \after \mathfrak{e} \after \eta
\\
& = &
B(b) \after BT(f) \after \mathfrak{e} \after \eta
\\
& = &
B(b) \after \mathfrak{e} \after TA(f) \after \eta
\\
& = &
B(b) \after \mathfrak{e} \after \eta \after A(f)
\\
& = &
\overline{\mathfrak{e}}_{(Y,b)} \after AU(f)
\end{array} \]

\noindent Apparently we don't need~\eqref{eqn:extensiondiagsingle}.

Conversely, given $\rho\colon AU \Rightarrow U\overline{B}$ we
obtain:
\[ \xymatrix@C+0.5pc{
TA(X)\ar[r]^-{TA(\eta)} & TAT(X)\ar[r]^-{T(\rho_{(TX,\mu)})} &
   TBT(X)\ar[r]^-{\kappa} & BT^{2}(X)\ar[r]^-{B(\mu)} & BT(X)
} \]

We check that this map, call it $\overline{\rho}_{X}$, makes
Diagram~\eqref{eqn:extensiondiagsingle} commute.
\[ \begin{array}{rcl}
B(\mu) \after \kappa \after T(\overline{\rho})
& = &
B(\mu) \after \kappa \after TB(\mu) \after T(\kappa) \after
   T^{2}(\rho) \after T^{2}A(\eta)
\\
& = &
B(\mu) \after BT(\mu) \after \kappa \after T(\kappa) \after
   T^{2}(\rho) \after T^{2}A(\eta)
\\
& = &
B(\mu) \after B(\mu) \after \kappa \after T(\kappa) \after
   T^{2}(\rho) \after T^{2}A(\eta)
\\
& = &
B(\mu) \after \kappa \after \mu \after T^{2}(\rho) \after T^{2}A(\eta)
\\
& = &
B(\mu) \after \kappa \after T(\rho) \after TA(\eta) \after \mu
\\
& = &
\overline{\rho} \after \mu
\end{array} \]

\noindent Strangely, here we don't use that $\rho$ is a map
of EM-algebras.
}
\end{proof}

The composed step $U\varepsilon \circledcirc \rho \colon TAU
\Rightarrow U\overline{B}$ gives a corecursive algebra, by applying
the step-induced algebra lifting $G_{U\varepsilon \circledcirc \rho}
\colon \Alg(\overline{B}) \rightarrow \Alg(TA)$ to the final
$\overline{B}$-coalgebra $((\Theta, a), \zeta)$, from 
Lemma~\ref{lm:final-coalg-em}.  We call this corecursive algebra
$\caem^A \colon TA(\Theta) \rightarrow \Theta$ to distinguish it from
$\caem \colon BT(\Theta) \rightarrow \Theta$.  It is given by:
\[ \xymatrix@R-1.7pc@C-1pc{
\caem^{A} \coloneqq \Big(TA(\Theta) = TAU(\Theta,a)\ar[r]^-{T(\rho)} &
   TU\overline{B}(\Theta,a)\ar@{=}[d]\ar[rr]^-{U(\varepsilon)} & &
   U\overline{B}(\Theta,a)\ar@{=}[d]
\\
& TB(\Theta)\ar[r]_-{\kappa} & BT(\Theta)\ar[r]_-{B(a)} & 
   B(\Theta)\ar[r]_-{\zeta^{-1}} & \Theta\Big)
} \]

\noindent This corecursive algebra gives semantics to
$TA$-coalgebras. It can be expressed in terms of the corecursive
$BT$-algebra $\caem$, making use of $\rho_2 \colon \free A \Rightarrow
\overline{B} \free$, as follows.

\begin{lemma}
\label{lm:corec-alg-ta}
We have $\caem^A = \caem \after U(\rho_2) \colon TA(\Theta)
\rightarrow \Theta$. Explicitly:
\[ \begin{array}{rcccl}
\big(TA(\Theta) \xrightarrow{\caem^A} \Theta\big) 
& = &
\big(TA(\Theta) \xrightarrow{U(\rho_2)} U\overline{B} \free (\Theta) = BT(\Theta) \xrightarrow{B(a)} B(\Theta) \xrightarrow{\zeta^{-1}}\Theta\big).
\end{array} \]
\end{lemma}

\begin{proof}
We describe $\caem \after U(\rho_2)$ as south-east and $\caem^A$ as
east-south in:
\[ \xymatrix@C-1.3pc@R-1.3pc{
\llap{$TA(\Theta)=\;$}U\free A U(\Theta,a)
   \ar[dd]_{U(\rho_2)}\ar[rr]^-{T(\rho_{1})} & & 
   U\free U\overline{B}(\Theta,a)\ar@{=}[r]\ar[dd]_{U(\varepsilon)} & 
   TB(\Theta)\ar[d]^{\kappa} 
\\
& & & BT(\Theta)\ar[dd]^{B(a)}
\\
U\overline{B}\free U(\Theta,a)\ar@{=}[d]\ar[rr]^-{U\overline{B}(\varepsilon)}
   & & U\overline{B}(\Theta,a)\ar@{=}[dr] 
\\
BT(\Theta)\ar[rrr]_-{B(a)} & & & B(\Theta)\ar[rr]_-{\zeta^{-1}} & & \Theta
} \]

\noindent The upper left square commutes by
Lemma~\ref{lm:mates-useful}. \QED
\end{proof}

\begin{example}\label{ex:nda-generative}
	We illustrate the situation with a simple example:
	non-deterministic automata, viewed as coalgebras of the form
	$c \colon X \rightarrow \Powfin(\Sigma \times X + 1)$.
	To this end, we instantiate the setting with $\Cat{C} = \Sets$, 
	$T = \Powfin$ the finite powerset monad, 
	and $A(X) = \Sigma \times X + 1$. Moreover, we let $B(X) = 2 \times X^\Sigma$. 
	Note that there is a difference between $\Powfin A$-coalgebras
	and $B \Powfin$-coalgebras, if $\Sigma$ is infinite: the former are finitely branching 
	non-deterministic automata (that is, finitely many successors)
	whereas the latter are image-finite non-deterministic automata (that is, finitely
	many successors for every alphabet letter). 
	
	The lifting $\overline{B} \colon \EM(\Powfin) \rightarrow \EM(\Powfin)$ of
	$B$ is given as in Example~\ref{ex:first-em} and Theorem~\ref{thm:em-automata}. 
	In particular, the corecursive algebra
	$$
	\caem \colon 2 \times (\Powfin(2^{\Sigma^*}))^\Sigma \rightarrow 2^{\Sigma^*}
	$$
	is given by $\caem(o,\varphi)(\varepsilon) = o$ and $\caem(o,\varphi)(aw) = \bigvee_{\psi \in \varphi(a)} \psi(w)$. 
	
	The relevant step $\rho \colon AU \Rightarrow U\overline{B}$
	is most easily given by $\rho_4 \colon A \Rightarrow U \overline{B} \free = B \Powfin$.
	On a component $X$, we  define $
	(\rho_4)_X \colon \Sigma \times X + 1 \rightarrow 2 \times (\Powfin X)^\Sigma
	$
	by 
	$$
	(\rho_4)_X(a,x) = \left( \bot,\lambda b. \begin{cases} \{x\} & \text{ if } a=b \\ \emptyset & \text{ otherwise } \end{cases} \right)\,, \quad (\rho_4)_X(*) = (\top, \lambda b. \emptyset) \,.
	$$
	Then $(\rho_2)_X \colon \Powfin(\Sigma \times X + 1) \rightarrow 2 \times (\Powfin X)^\Sigma$
	is the adjoint transpose, given by 
	$$
	\rho_2(S) = \left(
		\bigvee_{* \in S} \top, \lambda a. \{x \mid (a,x) \in S\}
	\right).
	$$
	This coincides with the extension law given in~\cite{JacobsSS15}. 
	
	By Lemma~\ref{lm:corec-alg-ta}, the corecursive $\Powfin A$-algebra obtained from the final $\overline{B}$-coalgebra is given by
	$\caem^A = \caem \after U(\rho_2) \colon \Powfin(\Sigma \times 2^{\Sigma^*} + 1) \rightarrow 2^{\Sigma^*}$,
	which is: 
	$$
	\caem^A(S)(\varepsilon) = \bigvee_{* \in S} \top \,, \qquad \caem^A(S)(aw) = \bigvee_{(a, \psi) \in S} \psi(w) \,.
	$$
	Given a coalgebra $c \colon X \rightarrow \Powfin(\Sigma \times X + 1)$, 
	the unique coalgebra-to-algebra morphism $\trem^A \colon X \rightarrow 2^{\Sigma^*}$ from $c$ to $\caem^A$
	is thus given by $\trem^A(x)(\varepsilon) = \bigvee_{* \in c(x)} \top$
	and $\trem^A(x)(aw) = \bigvee_{(a,y) \in c(x)} \trem^A(y)(w)$.	
\end{example}
For examples of extension laws for weighted and probabilistic automata, see~\cite{JacobsSS15}.

\section{Traces via Logic}\label{sec:logic}

This section illustrates how the `logical' approach to trace semantics
of~\cite{KlinR16}, ultimately based on the testing framework introduced in~\cite{PavlovicMW06}, fits in our
general framework.  In this approach, traces are viewed as logical
formulas, also called tests, which are evaluated for states. These
tests are obtained via an initial algebra of a functor $L$. The
approach works both for $TB$ and $BT$-coalgebras (and could, in
principle, be extended to more general combinations).  We start by
listing our assumptions in this section, and continue by showing
how these assumptions lead to a corecursive algebra giving trace 
semantics in the general framework of Section~\ref{sect:stepsem}.

\begin{assumption}[Traces via Logic] In this section, we assume:
\begin{enumerate}
\item \label{assump:adj} An adjunction $F\dashv G$ between categories
  $\cat{C} \leftrightarrows \op{\cat{D}}$.

\item \label{assump:modality} A functor $T$ on $\cat{C}$ with a
step $\tau \colon TG \Rightarrow G$. 
  %`modality' step, \textit{i.e.}, natural transformation $\tau\colon TG
%  \Rightarrow G$. %, as in Lemma~\ref{lem:monadmap}.

\item \label{assump:logic} 
  A functor  $B\colon\cat{C}
  \rightarrow \cat{C}$ and a functor $L\colon\cat{D}
  \rightarrow \cat{D}$ with a step $\delta \colon BG \Rightarrow GL$.

%  A `behaviour' functor $B\colon\cat{C}
%  \rightarrow \cat{C}$ and a `logical' functor $L\colon\cat{D}
%  \rightarrow \cat{D}$ with a `logic' in the form of a natural
%  transformation $\delta \colon BG \Rightarrow GL$ as in
%  Lemma~\ref{lem:logicnatural}.

\item \label{assump:formulas} An initial algebra $\alpha \colon
  L(\Phi) \congrightarrow \Phi$.
\end{enumerate}
\end{assumption}

\noindent We deviate from the convention of writing $\rho$ for `step',
since the above map $\tau$ gives rise to multiple steps $\tau
\circledcirc \delta$ and $\delta \circledcirc \tau$
in~\eqref{eqn:logicextensions} below, in the sense of
Definition~\ref{thm:mates}; here we use `delta' instead of `rho'
notation since it is common in modal logic.

\begin{example}
	We take $\cat{C} = \cat{D} = \Sets$, and $F,G$ both the contravariant powerset
	functor $2^{-}$. 
	Non-deterministic automata are obtained either as $BT$-coalgebras
	with $B(X) = 2 \times X^A$ and $T$ the finite powerset functor;
	or as $TB$-coalgebras, with $B(X) = A \times X + 1$ and $T$ again the finite powerset functor. 
	In both cases, $L$ is given by $L(X) = A \times X + 1$, which has the set of words $A^*$ as 
	carrier of an initial algebra. 
	The map $\tau \colon T2^- \Rightarrow 2^-$ is defined
	by $\tau_X(S)(x) = \bigvee_{\varphi \in S} \varphi(x)$, and intuitively models
	the existential choice in the semantics of non-deterministic automata. The step
	$\delta$ and the language semantics are defined later in this section. 
\end{example}
The assumptions are close to the general step-and-adjunction setting~\eqref{eq:adj-setting}. Here,
we have an opposite category on the right,
and instantiate $H$ to $TB$ or $BT$:
\begin{equation}
\xymatrix{
\cat{C}\ar@/^2ex/[rr]^-{F} \save !L(.5) \ar@(dl,ul)^{H} \restore & \bot & \op{\cat{D}} \ar@/^2ex/[ll]^-G \save !R(.5) \ar@(ur,dr)^{L} \restore 
}
\quad \text{ where } H = TB \text{ or } H=BT\,.
\end{equation}

\noindent Notice that our assumptions already include a step $\delta$
(involving $B,L$) and a step $\tau$, which we can compose to obtain
steps for the $TB$ respectively $BT$ case:
\begin{equation}
\label{eqn:logicextensions}
\begin{array}{rclcl}
\tau \circledcirc \delta
& \coloneqq &
\xymatrix{\Big(TBG\ar@{=>}[r]^-{T\delta} &
   TGL\ar@{=>}[r]^-{\tau L} & GL\Big) }
   & \quad\quad &
   \xymatrix{
   \op{\CoAlg(L)} \ar[r]^{G_{\tau \circledcirc \delta}}
   	& \Alg(TB)}
\\
\delta \circledcirc {\tau} 
& \coloneqq &
\xymatrix{\Big(BTG\ar@{=>}[r]^-{B\tau} & BG\ar@{=>}[r]^-{\delta} & GL\Big)}
   & &
   \xymatrix{
   \op{\CoAlg(L)} \ar[r]^{G_{\delta \circledcirc \tau}}
   	& \Alg(BT)}
\end{array}
\end{equation}

\noindent Both $\tau \circledcirc \delta$ and $\delta \circledcirc
\tau$ are steps, and hence give rise to step-induced algebra liftings
$G_{\tau \circledcirc \delta}$ and $G_{\delta \circledcirc \tau}$ of
$G$ (Section~\ref{sect:stepsem}).  By Theorem~\ref{thm:finalcor}, we
obtain two corecursive algebras by applying these liftings to the
inverse of the initial algebra, i.e., the (inverse of the) final
coalgebra in $\op{\cat{D}}$:
\begin{equation}
\label{diag:twocorecursivealgebras}
\begin{array}{rcccl}
\calogtb 
& \coloneqq & 
G_{\tau \circledcirc \delta}(\Phi,\alpha^{-1})
& = &
\xymatrix{\Big(TBG(\Phi)\ar[r]^-{\tau \circledcirc \delta} &
  GL(\Phi)\ar[r]^-{G(\alpha^{-1})}_-{\cong} & G(\Phi)\Big)}, \\
\calogbt 
& \coloneqq & 
G_{\delta \circledcirc \tau}(\Phi,\alpha^{-1})
& = &
\xymatrix{\Big(BTG(\Phi)\ar[r]^-{\delta \circledcirc \tau} &
  GL(\Phi)\ar[r]^-{G(\alpha^{-1})}_-{\cong} & G(\Phi)\Big)\,.
}
\end{array}
\end{equation}
These corecursive algebras define trace semantics for any
$TB$-coalgebra $(X,c)$ and $BT$-coalgebra $(Y,d)$:
\begin{equation}\label{eq:traces-logic}
\vcenter{\xymatrix@R-0.5pc@C+0.5pc{
	X \ar@{-->}[r]^{\trlog_c} \ar[d]_-c
		& G(\Phi) \\
	TB(X) \ar@{-->}[r]^-{TB(\trlog_c)}
		& TBG(\Phi) \ar[u]_-{\calogtb}
}}
\qquad\qquad
\vcenter{\xymatrix@R-0.5pc@C+0.5pc{
	Y \ar@{-->}[r]^{\trlog_d} \ar[d]_-d
		& G(\Phi) \\
	BT(Y) \ar@{-->}[r]^-{BT(\trlog_d)}
		& BTG(\Phi) \ar[u]_-{\calogbt}
}}
\end{equation}
It is instructive to characterise this trace semantics in terms of the transpose
and the step-induced coalgebra liftings $F^{\tau \circledcirc \delta}$ and $F^{\delta \circledcirc \tau}$,
showing how they arise as unique maps from an initial algebra:
\begin{equation}\label{eq:traces-logic-dual}
\vcenter{\xymatrix@R-0.5pc@C+1pc{
	F(X) 
		& \Phi \ar[d]^-{\alpha^{-1}}  \ar@{-->}[l]_-{\overline{\trlog_c}}  \\
	LF(X) \ar[u]^-{F^{\tau \circledcirc \delta}(X,c)} 
		& L(\Phi) \ar@{-->}[l]_-{L(\overline{\trlog_c})}
}}
\qquad\qquad
\vcenter{\xymatrix@R-0.5pc@C+1pc{
	F(Y) 
		& \Phi \ar[d]^-{\alpha^{-1}} \ar@{-->}[l]_-{\overline{\trlog_d}}  \\
	LF(Y) \ar[u]^-{F^{\delta \circledcirc \tau}(Y,d)} 
		& L(\Phi) \ar@{-->}[l]_-{L(\overline{\trlog_d})}
}}
\end{equation}

In the remainder of this section, we show two classes of examples of
the logical approach to trace semantics.  With these descriptions we retrieve most
of the examples from~\cite{KlinR16} in a smooth manner.

\begin{proposition}\label{prop:automaton-logic}
	Let $\Omega$ be a set, $T \colon \Sets \rightarrow \Sets$ a functor
	and $t \colon T(\Omega) \rightarrow \Omega$ a map. 
	Then the set of languages $\Omega^{A^*}$ carries a corecursive algebra
	for the functor $\Omega \times T(-)^A$.
	%The following algebra
%	is corecursive:
%	$$
%	\xymatrix@C=1.2cm{
%		\Omega \times (T(\Omega^{A^*}))^A \ar[r]^-{\idmap \times \st^A}
%			& \Omega \times ((T\Omega)^{A^*})^{A} \ar[r]^-{\idmap \times (t^{A^*})^A}
%			& \Omega \times (\Omega^{A^*})^A \ar[r]^-{\cong}
%			& \Omega^{A^*}
%	}
%	$$
%	where the isomorphism is a final $\Omega \times (-)^A$-coalgebra.
	Given a coalgebra $\langle o, f \rangle \colon X \rightarrow \Omega \times T(X)^A$,
	the unique coalgebra-to-algebra morphism $\trlog \colon X \rightarrow \Omega^{A^*}$
	satisfies
	\begin{align*}
		\trlog(x)(\varepsilon) = o(x) \qquad
		\trlog(x)(aw) = t\Big(T(\evmap_w \after \trlog)(f(x)(a))\Big)
	\end{align*}
	for all $x \in X$, $a \in A$ and $w \in A^*$. 
\end{proposition}

\begin{proof}
We instantiate the assumptions in the beginning of this section by
$\Cat{C} = \Cat{D} = \Sets$, $F = G = \Omega^{-}$, $B(X) = \Omega
\times X^A$, $L(X) = A \times X + 1$ and $T$ the functor
from the statement. The initial $L$-algebra is $\alpha \colon A\times A^* + 1
\congrightarrow A^*$.  The map $t$ extends to a modality $\tau
\colon TG \Rightarrow G$, given on components by
	$$
	\tau_X :=
	\big(
	\xymatrix{
		T(\Omega^X) \ar[r]^-{\st}
			& T(\Omega)^X \ar[r]^-{t^X}
			& \Omega^X
	}
	\big) \,.
	$$
	The logic $\delta \colon BG \Rightarrow GL$ is given by the isomorphism 
	$\Omega \times (\Omega^-)^A \cong \Omega^{(A \times -) + 1}$.
	Instantiating~\eqref{diag:twocorecursivealgebras} we obtain the corecursive $BT$-algebra
	$$
	\xymatrix@C=1.2cm{
		\Omega \times T(\Omega^{A^*})^A \ar[r]^-{\idmap \times (\st)^A}
			& \Omega \times (T(\Omega)^{A^*})^{A} \ar[r]^-{\idmap \times (t^{A^*})^A}
			& \Omega \times (\Omega^{A^*})^A \ar[r]^-{\Omega^{\alpha^{-1}} \after \delta}
			& \Omega^{A^*}
	} \, .
	$$
	
	The concrete description of $\trlog$ follows by
	spelling out the coalgebra-to-algebra diagram that characterises it.
	In particular, we have:
	\begin{align*}
		 \trlog(x)(aw) 
 		&= (\Omega^{\alpha^{-1}} \circ \delta_{A^*} \circ \idmap \times (t^{A^*} \circ \st \circ T(\trlog))^A \circ \langle o, f \rangle(x))(aw) \\
 		&= \delta_{A^*}(\idmap \times (t^{A^*} \circ \st \circ T(\trlog))^A \circ \langle o, f \rangle(x))(a,w) \\
 		&= ((t^{A^*} \circ \st \circ T(\trlog))^A \circ f(x))(a)(w) \\
 		&= (t^{A^*} \circ \st \circ T(\trlog)(f(x)(a)))(w) \\
 		&= t (\st \circ T(\trlog)(f(x)(a))(w)) \\
 		&= t (T(\evmap_w \circ \trlog)(f(x)(a)))
	\end{align*}
	for all $x \in X$, $a \in A$ and $w \in A^*$. 
	\QED
\end{proof}

\begin{example}
\label{ex:automata-traces}
We instantiate the trace semantics $\trlog$ from
Proposition~\ref{prop:automaton-logic} for various choices of
$\Omega$, $T$ and $t$. 
Similar to the instances in Example~\ref{ex:em-automata},
we consider a coalgebra $\langle o, f \rangle
\colon X \rightarrow \Omega \times T(X)^A$, and we always have
$\trlog(x)(\varepsilon) = o(x)$. 
The cases of non-deterministic automata (NDA-$\exists$, NDA-$\forall$)
and probabilistic automata (PA) are the same as in Example~\ref{ex:em-automata}.
However, in constrast to the Eilenberg-Moore approach
and other approaches to trace semantics, a monad structure
on $T$ is not required here. This is convenient as it also allows to treat 
alternating automata (AA), where $T = \Powfin \Powfin$; the latter does not
carry a monad structure~\cite{KlinS18}. 
$$
\begin{array}{c|c|c|c|c}
		& T & \Omega & t \colon T(\Omega) \rightarrow \Omega & \trlog(x)(aw) \\
		\hline 
		\text{NDA-}\exists
		& \Powfin
		& 2 = \{\bot,\top\}
		& S \mapsto \bigvee S
		& \bigvee_{y \in f(x)(a)} \trlog(y)(w) \\
		\text{NDA-}\forall
		& \Powfin
		& 2 = \{\bot,\top\}
		& S \mapsto \bigwedge S
		& \bigwedge_{y \in f(x)(a)} \trlog(y)(w) \\
		\text{PA}
		& \Dstfin
		& [0,1]
		& \varphi \mapsto \sum_{p \in [0,1]}p \cdot \varphi(p) 
		& \sum_{y \in X} \trlog(y)(w) \cdot f(x)(a)(y) \\
		\text{AA}
		& \Powfin\Powfin 
		& 2 = \{\bot,\top\}
		& S \mapsto \bigvee_{T \in S} \bigwedge_{b \in T} b
		&  \bigvee_{T \in f(x)(a)}\bigwedge_{y \in T} \trlog(y)(w)
\end{array}
$$
\end{example}

We also describe a logic for polynomial functors constructed from a
signature. Here, we model a signature by a functor $\Sigma \colon
\mathbb{N} \rightarrow \Sets$, where $\mathbb{N}$ is the discrete
category of natural numbers. This gives rise to a functor $H_\Sigma
\colon \Sets \rightarrow \Sets$ as usual by $H_\Sigma(X) =
\coprod_{n\in \mathbb{N}} \Sigma(n) \times
X^n$.  
The initial algebra of $H_\Sigma$
consists of closed terms (or finite node-labelled trees) over the
signature.

\begin{proposition}
\label{prop:tree}
Let $\Omega$ be a meet semi-lattice with top element $\top$ as well as
a bottom element $\bot$, let $T \colon \Sets \rightarrow \Sets$ be a
functor, and $t \colon T(\Omega) \rightarrow \Omega$ a map.  Let
$(\Phi,\alpha)$ be the initial $H_\Sigma$-algebra.  The set $\Omega^\Phi$ of
`tree' languages carries a corecursive algebra for
the functor $TH_\Sigma$.  Given a coalgebra $c
\colon X \rightarrow TH_\Sigma(X)$, the unique coalgebra-to-algebra
map $\trlog \colon X \rightarrow \Omega^\Phi$ is given by
	\begin{align*}
	&\trlog(x)(\sigma(u_1, \ldots, u_n))
	 	= t(T(m) \after c(x))\,, \text{where} \\
	 &m = \left( u \mapsto \begin{cases} \bigwedge_i \trlog(x_i)(u_i) & \text{ if }\exists x_1\ldots x_n .\, u=(\sigma,x_1, \ldots, x_n) \\ \bot & \text{otherwise} \end{cases} 
	 \right) : H_\Sigma(X) \rightarrow \Omega
	\end{align*}
	for all $x \in X$ and $\sigma(u_1, \ldots, u_n) \in \Phi$. 
\end{proposition}
\begin{proof}
We use $\Cat{C} = \Cat{D} = \Sets$, $F = G = \Omega^{-}$, $B = L =
H_\Sigma$. The map $t$ extends to a
modality $\tau \colon TG \Rightarrow G$ as in the proof of
Proposition~\ref{prop:automaton-logic}.  The logic $\delta \colon
H_\Sigma \Omega^- \Rightarrow \Omega^{H_\Sigma(-)}$ is:
	$$
	\delta_X(\sigma_1,\phi_1, \ldots, \phi_n)(\sigma_2,x_1, \ldots, x_m) = 
	\begin{cases}
	\bigwedge_i \phi_i(x_i) & \text{ if }\sigma_1 = \sigma_2 \\ \bot & \text{otherwise}
	\end{cases} 
	$$	
	The corecursive algebra $\calogtb$ is then given by:
	$$
	\xymatrix{
		TH_\Sigma(\Omega^\Phi) \ar[r]^-{T(\delta)}
			& T(\Omega^{H_\Sigma(\Phi)}) \ar[r]^-{\st}
			& T(\Omega)^{H_\Sigma(\Phi)}\ar[r]^-(0.7){t^{H_\Sigma(\Phi)}}
			& \Omega^{H_\Sigma(\Phi)} \ar[r]_-{\cong}^-{\Omega^{\alpha^{-1}}}
			& \Omega^{\Phi}
	}\,. 
	$$	
%	The explicit characterisation of $\trlog$ is a straightforward computation. \QED
Now, given a coalgebra $c \colon X \rightarrow TH_\Sigma(X)$, we
compute:
	\begin{align*}
		& \trlog(x)(\sigma(u_1, \ldots, u_n)) \\
		&= (\Omega^{\alpha^{-1}} \after t^{H_\Sigma \Phi} \after \st \after T(\delta_\Phi) \after TH_\Sigma(\trlog) \after c(x))(\sigma(u_1, \ldots, u_n)) \\
		&= t ((\st \after T(\delta_\Phi) \after TH_\Sigma(\trlog) \after c(x))(\alpha^{-1}(\sigma(u_1, \ldots, u_n)))) \\
		&= t ((\st \after T(\delta_\Phi) \after TH_\Sigma(\trlog) \after c(x))(\sigma,u_1, \ldots, u_n)) \\
		&= t (T(\evmap_{(\sigma,u_1, \ldots, u_n)})(T(\delta_\Phi) \after TH_\Sigma(\trlog) \after c(x))) \\		
		&= t (T(\evmap_{(\sigma,u_1, \ldots, u_n)} \after \delta_\Phi \after H_\Sigma(\trlog))(c(x))) 
	\end{align*}
%	Note that in the above computation, we overload the notation $\sigma(u_1, \ldots, u_n)$, seeing it both as an element
%	of $\Phi$ and $H_\Sigma(\Phi)$ via the isomorphism $\alpha^{-1}$. 
	To conclude, we analyse the map 
	$\evmap_{(\sigma,u_1, \ldots, u_n)} \after \delta_\Phi \after H_\Sigma(\trlog)$:
	\begin{align*}
		& \evmap_{(\sigma,u_1, \ldots, u_n)} (\delta_\Phi (H_\Sigma(\trlog)(u))) \\
		& = \delta_\Phi (H_\Sigma(\trlog)(u))(\sigma, u_1, \ldots, u_n) \\
		& = \begin{cases} \bigwedge_i \trlog(x_i)(u_i) & \text{ if }\exists x_1\ldots x_n .\, u=(\sigma,x_1, \ldots, x_n) \\ \bot & \text{otherwise} \end{cases}  \,.
	\end{align*}
	This coincides with $m$ in the statement of the proposition. 
	\QED
\end{proof}

\begin{example}
\label{ex:tree}
Given a signature $\Sigma$, a coalgebra $c \colon X \rightarrow
\Powfin H_\Sigma (X)$ is a (top-down) \emph{tree
  automaton}. With $\Omega = \{\bot,\top\}$ and $t(S) = \bigvee S$, 
  Proposition~\ref{prop:tree} gives: 
	$$
	\trlog(x)(\sigma(t_1, \ldots, t_n))=\top \text{ iff } \exists x_1 \ldots x_n . (\sigma,x_1, \ldots, x_n) \in c(x) \wedge \bigwedge_{1 \leq i \leq n} \trlog(x_i)(t_i)
	$$
	for every state $x \in X$ and tree $\sigma(t_1, \ldots, t_n)$. 
	This is the standard semantics of tree automata. It is easily
	adapted to \emph{weighted} tree automata, see~\cite{KlinR16}.
\end{example}
In both Example~\ref{ex:automata-traces} and Example~\ref{ex:tree}, 
the step-induced coalgebra lifting $F^{\delta \circledcirc \tau}$ (respectively $F^{\tau \circledcirc \delta}$) of the underlying logic 
corresponds to reverse determinisation, see~\cite{KlinR16,Rot16} for details. In particular, in Example~\ref{ex:tree} 
it maps a top-down tree automaton to the corresponding bottom-up tree
automaton.

\section{Traces via Kleisli}\label{sec:kl}

In this section we briefly recall the `Kleisli approach' to trace
semantics~\cite{HasuoJS07}, and cast it in our abstract framework.  It
applies to coalgebras for a composite functor $TA$, where $T$ is a
monad modelling the type of branching and $A$ is a functor. For
example, a coalgebra $X \to \mathcal{P}(\Sigma \times X +S)$ has an
associated map $X \to \mathcal{P}(\Sigma^{*} \times S)$ that sends a
state $x \in X$ to the set of its complete traces.  (Taking $S=1$,
this is the usual language semantics of a nondeterministic automaton.)
To fit this to our framework, the monad $T$ is $\mathcal{P}$ and the
functor $A$ is $(\Sigma \times -) + S$.  In general, the following
assumptions are used.

\begin{assumption}[Traces via Kleisli] In this section, we assume:
\begin{enumerate}
\item An endofunctor $A \colon \cat{C} \rightarrow \cat{C}$ with an
  initial algebra $\beta \colon A(\Psi) \congrightarrow \Psi$.

\item A monad $(T,\eta,\mu)$, with the standard adjunction $J \dashv
  U$ between categories $\cat{C} \leftrightarrows \Kl(T)$, where $J(X)
  = X$ and $U(Y) = T(Y)$.

\item An extension $\overline{A}$ of $A$, as below:
\begin{equation}\label{eq:kl-lift}
   \vcenter{
	\xymatrix@R-0.5pc{
		\Kl(T) \ar[r]^{\overline{A}}
			& \Kl(T)  \\
		\Cat{C} \ar[r]_{A} \ar[u]^J
			& \Cat{C} \ar[u]_J
	}
	}
   \end{equation}

\noindent or, equivalently, a $\Kl$-law $\lambda \colon AT \Rightarrow TA$.

\item $(\Psi, J(\beta^{-1}))$ is a final $\overline{A}$-coalgebra.

%   \item A final $\overline{B}$-coalgebra $(\Gamma,\xi)$. 
   \end{enumerate}
\end{assumption}

\noindent In the case that $A$ is the functor $(\Sigma \times -) + S$, its
initial algebra is carried by $\Sigma^{*}\times S$, and the canonical
$\Kl$-law is given at $X$ by
\[ \xymatrix{
\Sigma\times TX + S\ar[rrr]^-{[T\mathsf{inl} \after st_{\Sigma,X},
    T\mathsf{inr} \after \eta_{S}]} & & & T(\Sigma\times X+S)
} \]

\noindent A central observation for the Kleisli approach to traces is
that the fourth assumption holds under certain order enrichment
requirements on $\Kl(T)$, see~\cite{HasuoJS07}.  In particular, these
hold when $T$ is the powerset monad, the (discrete) sub-distribution
monad or the lift monad.

The above assumptions give rise to the following instance of our
setting~\eqref{eq:adj-setting}:
$$
\xymatrix@C-1pc{
\cat{C}\ar@/^2ex/[rr]^-{J} \save !L(.5) \ar@(dl,ul)^{TA} \restore & \bot & \Kl(T) \ar@/^2ex/[ll]^-U \save !R(.5) \ar@(ur,dr)^{\overline{A}} \restore 
& \mbox{\qquad with \quad}
{\begin{array}{l}
\rho\colon TAU \Longrightarrow U\overline{A} \; \mbox{ where } \rho_{X} =
\\
\big(TATX\xrightarrow{T(\lambda)} T^{2}AX \xrightarrow{\mu} TAX\big)
\end{array}}
}
$$

\noindent Similar to the $\EM$-case in Section~\ref{sec:em}, the map
of adjunctions is most easily given in terms of $\rho_4 \colon TA
\Rightarrow U\overline{A}J$ as the identity, using that $\overline{A}$
extends $A$.

We apply the step-induced algebra lifting $G_\rho \colon
\Alg(\overline{A}) \rightarrow \Alg(TA)$ to the inverse of the final
$\overline{A}$-coalgebra, and obtain a corecursive $TA$-algebra,
called $\cakl$:
\begin{equation}
\label{eqn:cakl}
\begin{array}{rcl} 
\Big(TAT(\Psi) \xrightarrow{\cakl} T(\Psi)\Big) 
& \coloneqq &
G_\rho(\Psi,J(\beta^{-1})^{-1})
\\
& = &
G_\rho(\Psi, J(\beta))
\\
& = &
\Big(TAT(\Psi)\xrightarrow{\!T(\lambda)\!} T^{2}A(\Psi) \xrightarrow{\mu}
    TA(\Psi)\xrightarrow{\!T(\beta)\!} T(\Psi)\Big)
\end{array}
\end{equation}

\noindent By Theorem~\ref{thm:finalcor}, this algebra is corecursive,
i.e., for every coalgebra $c \colon X \rightarrow TA(X)$, there is a
unique map $\trkl_c$ as below:
\begin{equation}
\label{diag:trkl}
\vcenter{\xymatrix@R-0.5pc{
	X \ar[d]_-c \ar@{-->}[r]^{\trkl_c} \ar[d]
		& T(\Psi) \\
	TA(X) \ar@{-->}[r]^{TA(\trkl_c)}
		& TAT(\Psi) \ar[u]_-{\cakl}
}}
\end{equation}

\noindent The trace semantics is exactly as in~\cite{HasuoJS07}, to
which we refer for examples. For later use we note the following.

\begin{lemma}
\label{lem:trkl}
The above map $\cakl \colon TAT(\Psi) \rightarrow T(\Psi)$ is a
map of Eilenberg-Moore algebras $\mu_{AT(\Psi)} \rightarrow
\mu_{\Psi}$.
\end{lemma}

\begin{proof}
This follows by an easy calculation:
\[ \begin{array}[b]{rcl}
\cakl \after \mu
\hspace*{\arraycolsep}=\hspace*{\arraycolsep}
T(\beta) \after \mu \after T(\lambda) \after \mu
& = &
T(\beta) \after \mu \after \mu \after T^{2}(\lambda)
\\
& = &
T(\beta) \after \mu \after T(\mu) \after T^{2}(\lambda)
\\
& = &
\mu \after T^{2}(\beta) \after T(\mu) \after T^{2}(\lambda)
\hspace*{\arraycolsep}=\hspace*{\arraycolsep}
\mu \after T(\cakl).
\end{array} \eqno{\QEDbox} \]
\end{proof}

\section{Partial Traces for Input/Output}\label{sec:partial}
\subsection{Introduction}

To illustrate the versatility of our framework, we show next that it underpins a trace example quite different from the previous ones, one that arises in programming language semantics and involves both input and output actions~\cite{BowlerLevyPlotkin:infinstrat}.

To avoid confusion, it must be noted that the word ``trace'' is used with a different meaning in the automata and semantics communities, as follows.
\begin{itemize}
\item In the automata literature and the previous sections, a ``trace'' ends in
  acceptance.  Semanticists would call this a ``complete trace''.
\item By contrast, in the semantics literature~\cite{RoscoeBrookesHoare:csp,JagadeesanPitcherRiely:aspects,Laird:genref,LassenLevy:nfbistypes,LevyStaton:transgames,BowlerLevyPlotkin:infinstrat} and this section, a ``trace'' need not end in acceptance.  For example, a program that prints \texttt{Hello} and then diverges (hangs) must be distinguished from one that simply diverges, even though---since neither terminates---neither has a complete trace.  Accordingly, the string \texttt{Hello} is said to be a ``trace'' of the former program (but not the latter), and so is each prefix.  Automata theorists would call these ``partial traces''.
\end{itemize}
This section applies our framework to traces of the second kind, but before doing that, we need two pieces of background.  The first (Section~\ref{sect:bgtrace}) explains that, in a transition system for I/O, a state's set of traces form a \emph{strategy}.  The second (Section~\ref{sect:stratfin}) characterizes the poset of all strategies as a final coalgebra.  This is a result that appeared in~\cite{BowlerLevyPlotkin:infinstrat}.

\subsection{Trace sets as strategies} \label{sect:bgtrace}

  The story begins by fixing a \emph{signature}, which consists of a set
$K$ of operations, and for each $k \in K$ a set $\arity{k}$ called its
\emph{arity}.  Each operation $k \in K$ is regarded as an output
message requesting input, and $\arity{k}$ as the set of acceptable
inputs.\footnote{Many-sorted signatures, in the guise of ``interaction structures'', are used for a similar purpose in~\cite{HancockHyvernat:progintbastop}.}  Accordingly, we use the functor:
\begin{equation}
 \begin{array}{rcccl}
X 
     & \mapsto &
                 \Pow\Big(\sum_{k \in K} X^{\arity{k}}\Big)
& = &
\Pow\Big(\sum_{k \in K}\prod_{i \in \arity{k}} X\Big)
.
\end{array} \label{eq:funcoi}
\end{equation}
\noindent A \emph{transition system} is a coalgebra $c \colon X \to
\pset(\sum_{k \in K}\prod_{i \in \arity{k}} X)$.  For such a system, a
state $x \in X$ represents a program that nondeterministically outputs
some $k \in K$, then pauses until it receives some $i \in \arity{k}$,
and then is in another state.  We write:
\[ x \stackrel{k}{\Longrightarrow} (y_i)_{i \in \arity{k}}
\qquad\mbox{for}\qquad
(k,(y_i)_{i \in \arity{k}}) \in c(x). \]

\noindent A \emph{play} is a finite or infinite sequence
$k_0,i_0,k_1,i_1,\ldots$, where $k_r \in K$ and $i_r \in \arity{k_r}$.  It is so called because it may be viewed  a play in a game of two players, called Proponent and Opponent, where each output is a Proponent-move and each input an Opponent-move.  (The game terminology is slightly misleading in that there is no notion      of winning, and play can continue forever.)  A play of even length is \emph{active-ending} and one of odd length is
\emph{passive-ending}.

A \emph{strategy} (more precisely: nondeterministic finite trace strategy) is a set $\sigma$ of passive-ending plays such that 
$sik \in \sigma$ implies $s \in \sigma$.  Again, this terminology is based on the game idea, as a strategy tells Proponent (nondeterminstically) how to play.  The poset of all strategies, ordered by inclusion ($\subseteq$), is written $\strat$.

Let $(X,c)$ be a transition system, and $x \in X$ a state.  A
passive-ending play $k_0,i_0,\ldots,k_n$ is said to be a \emph{trace}
of $x$ when there is a sequence
\[ x = x_0 \stackrel{k_0}{\Longrightarrow} (y^{0}_i)_{i \in \arity{k_0}} 
\; , \qquad
y^{0}_{i_{0}} = x_1 \stackrel{k_1}{\Longrightarrow}  (y^{1}_i)_{i \in \arity{k_1}}  
\; , \; \cdots \]

\noindent The set of all such traces forms a strategy.  Note that
active-ending traces need not be considered, since these are
determined by the passive-ending traces.  Infinite traces are not
considered in~\cite{BowlerLevyPlotkin:infinstrat}, nor are they here.
Conversely, every strategy can be obtained in this
way~\cite[Proposition 6.1]{BowlerLevyPlotkin:infinstrat}.

\subsection{Strategies form a final coalgebra}\label{sect:stratfin}
  
A \emph{complete semilattice} is a poset with all suprema.  Hence it
also has all infima, which allows it to be called a ``complete
lattice''.  Clearly the poset $\strat$ of all strategies, ordered by
inclusion ($\subseteq$), is a complete semilattice.  Let $\complatt$
be the category of complete semilattices and \emph{homomorphisms},
i.e., monotone functions that preserve suprema.  It was shown
in~\cite{BowlerLevyPlotkin:infinstrat} that $\strat$ is a final
coalgebra for a certain endofunctor on $\complatt$, which we shall
describe in several steps.

Firstly, an \emph{almost complete semilattice} is a poset where every
nonempty subset has a supremum.  Hence every lower-bounded subset has
an infimum, but binary meets $a \wedge b$ need not exist in general.
Let $\alcomplatt$ be the category of almost complete semilattices and
\emph{homomorphisms}, i.e., monotone functions that preserve suprema
of nonempty sets.  Informally, our motivation for using this category
is the fact that, up to trace equivalence, an I/O action such as
printing commutes with binary nondeterminism, and more generally with
$I$-ary nondeterministic choice for any nonempty set $I$.  This point
(and the special role of the empty set) is developed in more detail
in~\cite{BowlerLevyPlotkin:infinstrat}.

For any set $J$ we define two functors:
\[ \xymatrix@C+0.5pc{
\complatt^{J}\ar[r]^-{\prod_{J}} & \alcomplatt 
& 
\alcomplatt^{J}\ar[r]^-{\bigolift_{J}} & \complatt 
} \]

\noindent as follows. (In~\cite{BowlerLevyPlotkin:infinstrat} they are linked to universal properties.)
\begin{itemize}
                                              \item For a family $(A_j)_{j \in J}$ of complete semilattices, let $\prod_{j \in J}A_j$ be the cartesian product.  Endowed with pointwise order, it is an almost complete (in fact complete) semilattice.                                                            
 \item For a family $(f_{j} \colon A_{j} \to B_{j})_{j \in J}$ of complete semilattice homomorphisms, let $\prod_{j \in J}f_j \colon \prod_{j \in J}A_j \to \prod_{j \in J} B_j$ be the map sending $(a_j)_{j \in J}$ to $(f_ja_j)_{j \in J}$.                        
  \item For a family  $(A_j)_{j \in J}$ of almost complete semilattices, let $\bigolift_{j \in J}A_j$ be the set of pairs $(U,(a_j)_{j \in U})$ where $U \in \pset J$ and $a_j \in A_j$ for all $j \in U$.  It is a complete semilattice when endowed with the following order: we have $(U,(a_j)_{j \in U}) \leqslant (V,(b_j)_{j \in V})$ when $U \subseteq V$ and $a_j \leqslant b_j$ for all $j \in U$.
\item For a family $(f_{j} \colon A_{j} \to B_{j})_{j \in J}$ of almost complete semilattice homomorphisms, let $\bigolift_{j \in J}f_j \colon \bigolift_{j \in J}A_j \to \bigolift_{j \in J} B_j$ be the map sending a pair $(U,(a_j)_{j \in U})$ to $(U,(f_ja_j)_{j \in U})$.
  \end{itemize}
  From these we build our endofunctor
  \begin{displaymath}
    {\bigoplus_{k \in K}}^{\bot}\prod_{i \in \arity{k}} \colon \complatt \to \complatt
  \end{displaymath}
whose final coalgebra is given as follows.
 \begin{theorem}{\rm \cite[Theorem 6.3]{BowlerLevyPlotkin:infinstrat}} \label{thm:finalcoalgstrat}
 Let $\parstrat \colon \strat \to \bigolift_{k \in K}\prod_{i \in \arity{k}} \strat$ send a strategy $\sigma$ to $(\initsigma, ((\sigmaki_{i})_{i \in \arity{k}})_{k \in \initsigma})$, where
  \begin{eqnarray*}
    \initsigma & \eqdef & \setbr{k \in K \mid (k) \in \sigma} \\
    \sigmaki & \eqdef &  \setbr{s \mid k.i.s \in \sigma}
  \end{eqnarray*}
Then $(\strat, \parstrat)$ is a final $\bigolift_{k \in K}\prod_{i \in \arity{k}}$-coalgebra. 
\end{theorem}

\subsection{The step} \label{sect:iostep}

With the background completed, we now want to instantiate our general setting to form an account of traces.
% Note first that the forgetful functor  $U \colon \complatt \to \Sets$
%   is monadic, where the free complete semilattice on a set $X$ is
%   $\pset X$, ordered by inclusion ($\subseteq$), with unit $X \to \pset X$ sending
%   $x \mapsto \{x\}$.  Likewise the  forgetful functor $U \colon \alcomplatt \to \Sets$ is monadic, where
%   the free almost complete semilattice on a set $X$ is the set
%   $\psetplus X$ of nonempty subsets, ordered by inclusion ($\subseteq$), with unit
%   $X \to \psetplus X$ sending $x \mapsto \{x\}$.
Our adjunction and endofunctors are as follows:
\begin{displaymath}
    \xymatrix{
    \Sets \ar@/^{2ex}/[rr]^-{\pset} \save !L(.5) \ar@(dl,ul)^{\pset\sum_{k \in K} \prod_{i \in \arity{k}}} \restore & \bot & \complatt \ar@/^{2ex}/[ll]^-{U} \save !R(.5) \ar@(ur,dr)^{\bigolift_{k \in K}\prod_{i \in \arity{k}}} \restore 
% & \mbox{\qquad with} &
% HG \ar@{=>}[r]^-{\rho} & GL
}
\end{displaymath}
Here $U \colon \complatt \to \Sets$ is the forgetful functor, which is monadic.  Explicitly, the free complete semilattice on a set $X$ is
 $\pset X$, ordered by inclusion ($\subseteq$), with unit $X \to \pset X$ sending
   $x \mapsto \{x\}$.  Likewise the forgetful functor $U \colon \alcomplatt \to \Sets$ is monadic.  Explicitly, the free almost complete semilattice on a set $X$ is the set
  $\psetplus X$ of nonempty subsets, ordered by inclusion ($\subseteq$), with unit
  $X \to \psetplus X$ sending $x \mapsto \{x\}$.

  Our step is formulated using bimodules and 2-cells, which are explained in the Appendix.  Any functor $U \colon \catd\to \catc$ gives rise to a bimodule  $\rbim{U} \colon \catc \relto \catd$ by  Definition~\ref{def:lbimrbim}(\ref{item:rbim}), and then, for any set $J$, to a bimodule $(\rbim{U})^{J} \colon \catc^{J} \relto \catd^{J}$ by Definition~\ref{def:prodbim}.  Central to our story are the following 2-cells (in the sense of Definition~\ref{def:mapbim}(\ref{item:tcell})) defined for any set $J$.
 \begin{displaymath}
    \xymatrix{
      \Sets^{J} \ar[rr]|-@{|}^-{(\rbim{U})^{J}} \ar[d]_{\prod} \ar@{}[drr]|{\prod\ \Downarrow}
      & & \complatt^{J}  \ar[d]^{\prod}
      &
\Sets^{J} \ar[d]_{\pset \sum} \ar@{}[drr]|{\sumstar\ \Downarrow}  \ar[rr]|-@{|}^-{(\rbim{U})^{J}}
& & \alcomplatt^{J} \ar[d]^{\bigolift} \\
      \Sets  \ar[rr]|-@{|}_-{\rbim{U}} & & \alcomplatt &
      \Sets \ar[rr]|-@{|}_-{\rbim{U}} & & \complatt 
      }
    \end{displaymath}
They are defined as follows.
\begin{itemize}
\item Given a family of functions $(f_j \colon X_j \to B_j)_{j \in J}$, where
  $X_j$ is a set and $B_j$ a complete semilattice, the function
  $\prod_{j \in J}f_j \colon \prod_{j \in J}X_j \to \prod_{j \in J}
  B_j$ sends $(x_j)_{j \in J}$ to $(fx_j)_{j \in J}$.
  \item Given a family of functions  $(f_j \colon X_j \to A_j)_{j \in J}$, where
    $X_j$ is a set and $A_j$ an almost complete semilattice, the function
    $\sumstar_{j \in J} f_j \colon \pset \sum_{j \in J} X_j \to \bigolift_{j \in J}A_j$ sends $R$ to  $(L,(y_j)_{j \in L})$ where
    \begin{eqnarray*}
L & = &  \setbr{j \in J \mid \exists x \in X_j.\,\mathsf{in}_j\,x \in R} \\
y_j & = & \bigvee_{x \in X_j \,:\, \mathsf{in}_j\,x \in R} f_j(x) \betwixt \text{ for $j \in L$.}
    \end{eqnarray*}
\end{itemize}
Note that, as in Sections~\ref{sec:em}
and~\ref{sec:kl}, the $\rho_4$ version of 
$\sumstar$ is an isomorphism, namely:
\[ \xymatrix@R-2pc{
  \bigolift_{j \in J} \psetplus(X_j) \ar[r]^-{\cong} & 
   \pset\big(\sum_{j \in J} X_j\big) 
\\
(U,(Y_j)_{j \in U})\ar@{|->}[r] & \{(j,x) \mid j \in U, x \in Y_j\}
} \]

\noindent Combining these 2-cells, we obtain the following 2-cell:
\begin{equation} \label{eqn:twocellforoi}
  \xymatrix{
    \Sets \ar@{}[drrr]|{\sumstar_{k \in K}\prod_{i \in \arity{k}} \ \Downarrow} \ar[d]_{\pset\sum_{k \in K} \prod_{i \in \arity{k}}}  \ar[rrr]|-@{|}^{\rbim{U}} & & & \complatt \ar[d]^{\bigolift_{k \in K}\prod_{i \in \arity{k}}} \\
    \Sets \ar[rrr]|-@{|}_{\rbim{U}} & & & \complatt
    }
\end{equation}
As Theorem~\ref{thm:matesbim} explains, this provides our step
\begin{displaymath}
\rho \colon \pset\sum_{k \in K} \prod_{i \in \arity{k}}U \Rightarrow U \bigoplus_{k \in K} \!{}^{\bot}\!\!\prod_{i \in \arity{k}}
\end{displaymath}
.  % To apply Proposition~\ref{prop:corec-bimod}, we require a final coalgebra for $\bigolift_{k \in K}\prod_{i \in \arity{k}}$, which is obtained as follows.
% \begin{theorem}{\rm \cite[Theorem 6.3]{BowlerLevyPlotkin:infinstrat}}
%   Let $\strat$ be the complete semilattice of strategies, ordered by inclusion.   Let $\parstrat \colon \strat \to \bigolift_{k \in K}\prod_{i \in \arity{k}} \strat$ send a strategy $\sigma$ to $(\initsigma, ((\sigmaki_{i})_{i \in \arity{k}})_{k \in \initsigma})$, where
%   \begin{eqnarray*}
%     \initsigma & \eqdef & \setbr{k \in K \mid (k) \in \sigma} \\
%     \sigmaki & \eqdef &  \setbr{s \mid k.i.s \in \sigma}
%   \end{eqnarray*}
% Then $(\strat, \parstrat)$ is a $\bigolift_{k \in K}\prod_{i \in \arity{k}}$-final coalgebra. 
% \end{theorem}

From Theorem~\ref{thm:finalcoalgstrat} with
Proposition~\ref{prop:corec-bimod}(\ref{item:corec-bimod-left}), we
see that, for every coalgebra $c \colon X \to \pset\sum_{k \in K}
\prod_{i \in \arity{k}}X$, there is a unique morphism to
$(\strat,\parstrat)$.  Specifically~\cite[Theorem
  6.6]{BowlerLevyPlotkin:infinstrat} tells us that what this morphism
sends $x \in X$ to its set of traces.  Finally by
Proposition~\ref{prop:corec-bimod}(\ref{item:corec-bimod-right}),
$U_{\rho}(\strat,\parstrat)$ is corecursive, and the map from $(X,c)$
to it is the same, i.e., it sends $x \in X$ to its set of traces.

Note that, as in Section~\ref{sec:em}, we can use $\pset^{\rho}$ to
determinise a transition system $(X,c)$.  This is applied
in~\cite[Section 6.2]{BowlerLevyPlotkin:infinstrat} to obtain a
bisimulation method for trace equivalence.

\subsection{Input, then output}

We adapt the story above to use instead of (\ref{eq:funcoi}) the functor
\[ \begin{array}{rcccl}
X 
     & \mapsto &
             \prod_{k \in K} \pset (\arity{k} \times X )  
& = &
\prod_{k \in K} \pset \sum_{i \in \arity{k}}X  .
\end{array} \]
Now a \emph{transition system} is a coalgebra $c \colon X \to \prod_{k \in K} \pset \sum_{i \in \arity{k}}X$.  In this case, the behaviour of a state $x \in X$ is to first input $k \in K$ and then nondeterministically output some $i \in \arity{k}$, resulting in a new state $x'$.  We write
\[ x@k \stackrel{i}{\Longrightarrow} x'
\qquad\mbox{for}\qquad
(i,x') \in (c(x))_k. \]
Accordingly, the definitions of play, strategy and trace in Section~\ref{sect:bgtrace} are adjusted as follows.
\begin{itemize}
\item A \emph{play} is a finite or infinite sequence
  $k_0,i_0,k_1,i_1,\ldots$, where $k_r \in K$ and $i_r \in \arity{k_r}$.  A play of odd length is \emph{active-ending} and one of even length is \emph{passive-ending}.  
\item A \emph{strategy} is a set  $\sigma$ of passive-ending plays such that $\varepsilon \in \sigma$ (where $\varepsilon$ is the empty play) and 
  $ski \in \sigma$ implies $s \in \sigma$.
\item Let $(X,c)$ be a transition system, and $x \in X$ a state.  A
passive-ending play $k_0,i_0,\ldots,k_n,i_{n}$ is said to be a \emph{trace}
of $x$ when there is a sequence
\[ x = x_0 \;\,\qquad x_0@k_0 \stackrel{i_0}{\Longrightarrow} x_1
  \; , \qquad
 x_1@k_1 \stackrel{i_1}{\Longrightarrow} x_2
\; , \; \cdots \]
\noindent The set of all such traces form a strategy.  
\end{itemize}
The poset $\strat$ of all strategies, ordered by inclusion $(\subseteq)$, forms an almost complete (in fact complete) semilattice.  We adjust Theorem~\ref{thm:finalcoalgstrat} to say that $\strat$ carries a final coalgebra for the endofunctor $\prod_{k \in K} \bigolift_{i \in \arity{k}}$ on $\alcomplatt$.

Finally, we have the same results as in Section~\ref{sect:iostep}, but instead of (\ref{eqn:twocellforoi}) we use the following 2-cell:
\begin{displaymath}
  \xymatrix{
    \Sets \ar@{}[drrr]|{\prod_{k \in K}\sumstar_{i \in \arity{k}} \ \Downarrow} \ar[d]_{\prod_{k \in K} \pset \sum_{i \in \arity{k}}}  \ar[rrr]|-@{|}^{\rbim{U}} & & & \alcomplatt \ar[d]^{\prod_{k \in K}\bigolift_{i \in \arity{k}}} \\
    \Sets \ar[rrr]|-@{|}_{\rbim{U}} & & & \alcomplatt
    }
  \end{displaymath}
  To summarize, we first told our story for transition systems with ``active'' states, that output and then input.  (These systems are sometimes called ``generative''.)  In this section, we have adapted it for systems with ``passive'' states, that input and then output.  (These systems are sometimes called ``reactive''.)  Another variation would be transition systems with both active and passive states, as in~\cite{LevyStaton:transgames}.

\section{Comparison}
\label{sec:comp}

The presentation of trace semantics in terms of corecursive algebras
allows us to compare the different approaches by constructing algebra
morphisms between them. In three subsections, we compare the
Eilenberg-Moore approach with the logical approach, the Kleisli
approach with the logical approach, and finally we compare the Kleisli
and Eilenberg-Moore approaches.
%, by fitting the setting
%of~\cite{JacobsSS15}---which does not use corecursive algebras---in 
%the current framework.

\subsection{Eilenberg-Moore and Logic}
\label{sec:comp-em-logic}

To compare the Eilenberg-Moore approach with the logical approach, we
combine their assumptions, as follows.

\begin{assumption}[Comparison Eilenberg-Moore and Logic] In this subsection, we assume
an adjunction $F \dashv
G$, endofunctors $B,L$ and a monad $T$ as follows:
$$
\xymatrix{
\op{\cat{D}} \ar@/_2ex/[rr]_-G \save !L(.5) \ar@(dl,ul)^{L} \restore 
& \bot & \ar@/_2ex/[ll]_-{F} 
\cat{C}\ar@/^2ex/[rr]^-{\free} \ar@(ul,ur)^{BT} & \bot & \EM(T) \ar@/^2ex/[ll]^-U \save !R(.5) \ar@(ur,dr)^{\overline{B}} \restore 
}
$$

\noindent together with:
\begin{enumerate}
\item A final $B$-coalgebra $\zeta \colon \Theta \congrightarrow B(\Theta)$.

\item An $\EM$-law $\kappa \colon TB \Rightarrow BT$, or equivalently,
  a lifting $\overline{B}$ of $B$.

\item An initial algebra $\alpha \colon L(\Phi) \congrightarrow \Phi$.

\item A step $\delta \colon BG \Rightarrow GL$.
\item A step $\tau\colon TG \Rightarrow G$, whose components are
  $\EM$-algebras (a \emph{monad action}).
\end{enumerate}
\end{assumption}

\noindent 
Here we have assumed slightly more than the union of the assumptions of the two approaches.
The step $\tau$ is an assumption of the logical
approach in Section~\ref{sec:logic}, but there the compatibility with
the monad structure was not assumed---simply because $T$ was not
assumed to be a monad before. 
Here, we use this assumption as a first
compatibility requirement between the logical and Eilenberg-Moore approaches.

We note that $\tau$ being a monad
action is the same thing as $\tau$ being an $\EM$-law, involving the
monad $T$ on the left and the identity monad on the right. The next
result is therefore an instance of Theorem~\ref{thm:mates}.

\begin{lemma}\label{lm:eq-tau}
The following are equivalent:
\begin{enumerate}
\item a monad action $\tau_1 \colon TG \Rightarrow G$;

\item a map $\tau_2 \colon F \Rightarrow FT$, satisfying the obvious
  dual action equations;

\item a monad morphism $\tau_4 \colon T \Rightarrow GF$;

\item an extension $\widehat{F} \colon \Kl(T) \rightarrow \op{\cat{D}}
  \;(\,= \Kl(\Id))$ of $F$. 

\item a lifting $\widehat{G} \colon \op{\cat{D}} \rightarrow \EM(T)$
  of $G$. \QED
\end{enumerate}
\end{lemma}

\noindent Such monad actions and the corresponding liftings are used,
e.g., in~\cite{HinoKHJ16,Jacobs17b,Hasuo15} where $\widehat{F}$ is
called $\Pred$. We use $\widehat{\,\cdot\,}$ to indicate liftings
associated with the step $\tau$, in order to create a distinction with
the lifting $\overline{\,\cdot\,}$ associated with $\kappa$.

We now start focusing on the actual comparison between the
Eilenberg-Moore and logical approach.  First, observe that the step
$\delta \colon BG \Rightarrow GL$ gives a lifting $G_{\delta} \colon
\Alg(L) \rightarrow \Alg(B)$, where $G$ is a functor $\op{\cat{D}}
\rightarrow \cat{C}$. The `opposite' requires some care: the initial
algebra $\alpha \colon L(\Phi) \rightarrow \Phi$ in $\cat{D}$ forms a
final coalgebra $\alpha \colon \Phi \rightarrow L(\Phi)$ in
$\op{\cat{D}}$, and thus a corecursive algebra $\alpha^{-1} \colon
L(\Phi) \rightarrow \Phi$ in $\op{\cat{D}}$.  Hence, applying
$G_{\delta}$ to the latter corecursive $L$-algebra gives a corecursive
$B$-algebra, namely:
\[ \xymatrix@C-1.3pc{
G_{\delta}(\Phi, \alpha^{-1}) \coloneqq 
\Big(BG(\Phi)\ar[rr]^-{\delta} & & GL(\Phi)\ar[rrr]^-{G(\alpha^{-1})} 
   & & & G(\Phi)\Big).
} \]

\noindent Since this algebra is corecursive we obtain a unique map
$\emlog$ as in the following diagram:
\begin{equation}\label{eq:theory-finalb}
\vcenter{
\xymatrix@R-0.5pc{
\Theta \ar[dd]_{\zeta}^{\cong} \ar[rr]^{\emlog} & & G(\Phi) 
\\
& & GL(\Phi) \ar[u]_{G(\alpha^{-1})}
\\
B(\Theta) \ar[rr]^{B(\emlog)} & & BG(\Phi)\ar[u]_{\delta}
}}
\end{equation}

\noindent This $\emlog \colon \Theta \rightarrow G(\Phi)$ is a
morphism from the carrier of the corecursive algebra $\caem \colon
BT(\Theta) \rightarrow \Theta$, from the Eilenberg-Moore
approach~\eqref{eqn:caem}, to the carrier of the corecursive algebra
$\calogbt \colon BTG(\Phi) \rightarrow G(\Phi)$, from the logical
approach~\eqref{diag:twocorecursivealgebras}. Note that, by the above
diagram, $\emlog$ is a $B$-algebra morphism, whereas $\caem$ and
$\calogbt$ are $BT$-algebras.  We next describe a sufficient condition
under which the map $\emlog$ is a $BT$-algebra morphism from $\caem$
to $\calogbt$, which implies that the logical trace semantics factors
through the Eilenberg-Moore trace semantics, see subsequent
Theorem~\ref{thm:logic-and-em}.

\begin{lemma}\label{lm:condition-em-log}
The $\EM$-law $\kappa\colon TB\Rightarrow BT$ commutes with the step
compositions in~\eqref{eqn:logicextensions}, as in:
	\begin{equation}\label{eq:kappa-rho}
	\vcenter{
	\xymatrix@R-0.5pc{
		TBG \ar[rr]^{\kappa G} \ar[dr]_-{\tau \circledcirc \delta}
			& & BTG \ar[dl]^-{\delta \circledcirc \tau} \\
			& GL & 
	}
	}
	\end{equation}

\noindent iff there is a natural transformation $\widehat{\delta}
\colon \overline{B}\widehat{G} \Rightarrow \widehat{G}L$ satisfying
$U(\widehat{\delta}) = \delta$ in:
\[ \xymatrix@R-0.5pc{
& & \EM(T)\ar[dd]_{U}\save !R(.5) \ar@(ur,dr)^{\overline{B}} \restore 
  & \quad\mbox{with}\quad
  \overline{B}\widehat{G} \ar@{=>}[r]^-{\widehat{\delta}} & \widehat{G}L
\\
\op{\cat{D}}\ar[urr]^-{\widehat{G}}\ar[drr]_{G}
   \save !L(.5) \ar@(dl,ul)^{L} \restore 
\\
& & \cat{C}\save !R(.5) \ar@(ur,dr)^{B} \restore 
  & \quad\mbox{with}\quad
  BG \ar@{=>}[r]^-{\delta} & GL
} \]

\noindent The functor $\widehat{G} \colon \op{\cat{D}} \rightarrow
\EM(T)$ is the lifting corresponding to $\tau$, see
Lemma~\ref{lm:eq-tau}.
\end{lemma}

\begin{proof}
The existence of such a $\widehat{\delta}$ amounts to the property that each
component $\delta_X \colon BG(X) \rightarrow GL(X)$ is a $T$-algebra
homomorphism from $\overline{B}\widehat{G}(X)$ to $\widehat{G}L(X)$,
i.e., the following diagram commutes:
$$
\xymatrix@R-0.8pc@C=1.3cm{
TBG(X) \ar[r]^{T\delta} \ar[d]_{\kappa}
	& TGL(X) \ar[dd]^{\tau} \\
BTG(X) \ar[d]_{B(\tau)} & \\
	BG(X) \ar[r]^{\delta}
		& GL(X)
}
$$ 

\noindent This corresponds exactly to~\eqref{eq:kappa-rho},
see~\eqref{eqn:logicextensions}.  \QED
\end{proof}

\begin{theorem}	
\label{thm:logic-and-em}
If the equivalent conditions in Lemma~\ref{lm:condition-em-log} hold, then 
the map $\emlog$ defined in~\eqref{eq:theory-finalb} is a $BT$-algebra morphism 
from $\caem$ to $\calogbt$, as on the left below.
	$$
	\xymatrix@R-0.5pc@C=1.3cm{
		BT(\Theta) \ar[d]_-{\caem} \ar[r]^-{BT(\emlog)}
			& BTG(\Phi) \ar[d]^-{\calogbt} \\
		\Theta \ar[r]_-{\emlog} 
			& G(\Phi)  
	}
	\qquad\qquad
	\xymatrix@R-0.5pc{
		& X \ar[dl]_-{\trem_c} \ar[dr]^(0.52){\trlog_c} & \\
		\Theta \ar[rr]_{\emlog}& &  G(\Phi)
	}
	$$ 

\noindent In that case, for any coalgebra $\smash{X
  \stackrel{c}{\rightarrow} BT(X)}$ the triangle on the right
commutes.
%	Moreover, if $\delta$ is expressive, then $\emlog$ is monic. 
%	Finally, if $\delta$ is an isomorphism, then $\emlog$ is an isomorphism as well. 
\end{theorem}

\begin{proof}
We use that $\caem = \zeta^{-1} \after B(a) \colon BT(\Theta)
\rightarrow \Theta$, where $((\Theta,a),\zeta)$ is the final
$\overline{B}$-coalgebra, see Section~\ref{sec:em}.  We need to prove
that the outer rectangle of the following diagram commutes.
$$\xymatrix@R-0.5pc{
BT(\Theta) \ar[r]^-{B(a)} \ar[d]_{BT(\emlog)}
  & B(\Theta)\ar[d]^-{B(\emlog)}\ar[rr]^-{\zeta^{-1}}_-{\cong}
			& 
			& \Theta \ar[d]^-{\emlog} \\
		BTG(\Phi) \ar[r]_-{B(\tau_{1})}\ar`d[r]`[rrru]_{\calogbt}[rrr]
			& BG(\Phi) \ar[r]_-{\delta}
			& GL(\Phi) \ar[r]_-{G(\alpha^{-1})}^-{\cong}
			& G(\Phi)
	}
	$$
	
\noindent The rectangle on the right commutes by definition of
$\emlog$. For the square on the left, it suffices to show $\emlog
\after a = \tau_{1} \after T(\emlog)$; this is equivalent to $F(a)
\after \overline{\emlog} = \tau_{2} \after \overline{\emlog}$ in:
\[ \xymatrix@C+2pc{
\Phi\ar[r]^-{\overline{e} = F(e) \after \varepsilon} & 
   F(\Theta)\ar@<+0.2pc>[r]^-{F(a)}\ar@<-0.2pc>[r]_-{\tau_{2}} & FT(\Theta)
} \]

\noindent Indeed, by transposing we have on the one hand:
\[ \begin{array}{rcccccl}
\overline{\emlog \after a}
& = &
F(a \after \emlog) \after \varepsilon
& = &
F(a) \after F(\emlog) \after \varepsilon
& = &
F(a) \after \overline{\emlog}.
\end{array} \]

\noindent And on the other hand, using that $\tau_2 \circ \varepsilon = F(\tau_1) \circ \varepsilon$ by Lemma~\ref{lm:mates-useful}:
\begin{align*}
	\tau_{2} \circ \overline{\emlog}
		&= \tau_2 \circ F(\emlog) \circ \varepsilon \\
		&= FT(\emlog) \circ \tau_2 \circ \varepsilon \\
		&= FT(\emlog) \circ F(\tau_1) \circ \varepsilon \\
		&= \overline{\tau_1 \circ T(\emlog)}
\end{align*}
%\noindent And on the other hand, using that $\tau_{2} = F(\tau_{1}
%\after T(\eta)) \after \varepsilon$,
%\[ \begin{array}{rcl}
%\tau_{2} \after \overline{e}
%& = &
%F(\tau_{1} \after T(\eta)) \after \varepsilon \after F(\emlog) \after \varepsilon
%\\
%& = &
%F(\tau_{1} \after T(\eta)) \after FG(F(\emlog) \after \varepsilon) \after 
%   \varepsilon
%\\
%& = &
%F\big(G(F(\emlog) \after \varepsilon) \after \tau_{1} \after T(\eta)\big) 
%   \after \varepsilon
%\\
%& = &
%F\big(\tau_{1} \after TG(F(\emlog) \after \varepsilon) \after T(\eta)\big) 
%   \after \varepsilon
%\\
%& = &
%F\big(\tau_{1} \after T(G(\varepsilon) \after GF(\emlog) \after \eta)\big) 
%   \after \varepsilon
%\\
%& = &
%F(\tau_{1} \after T(\emlog)) \after \varepsilon
%\\
%& = &
%\overline{\tau_{1} \after T(\emlog)}.
%\end{array} \]

\noindent By transposing the map $\emlog$ in~\eqref{eq:theory-finalb},
it follows that $\overline{\emlog} \colon \Phi \rightarrow F(\Theta)$
is the unique morphism from the initial $L$-algebra $\alpha \colon
L(\Phi) \congrightarrow \Phi$ to $F(\zeta) \after \delta_{2} \colon
LF(\Theta) \rightarrow F(\Theta)$. Hence, for the desired equality
$F(a) \after \overline{\emlog} = \tau_2 \after \overline{\emlog}$, it
suffices to prove that $F(a)$ and $\tau_2$ are both algebra
homomorphisms from $F(\zeta) \after \delta_2$ to a common algebra,
which in turn follows from commutativity of the following diagram.
$$\xymatrix@R-0.5pc{
LF(\Theta) \ar[r]^-{L(\tau_{2})} \ar[dd]_{\delta_2}
			& LFT(\Theta) \ar[d]^{\delta_2}
			& LF(\Theta)\ar[d]^{\delta_2} \ar[l]_{LF(a)}\\
			& FBT(\Theta) \ar[d]^{F(\kappa)}
			& FB(\Theta) \ar[dd]^{F(\zeta)} \ar[l]_{FB(a)}\\
		FB(\Theta) \ar[d]_{F(\zeta)} \ar[r]^-{\tau_2}
			& FTB(\Theta) \ar[d]^{FT(\zeta)}
			& \\
		F(\Theta) \ar[r]_-{\tau_2}
			& FT(\Theta)
			& F(\Theta) \ar[l]^-{F(a)}
	}
	$$

\noindent Using the translation $(-)_{1} \leftrightarrow (-)_{2}$ of
Theorem~\ref{thm:mates}, one can show that the upper-left rectangle is
equivalent to the assumption~\eqref{eq:kappa-rho}. To see this, we use
Lemma~\ref{lm:mates-composition} to obtain $(\delta \circledcirc
\tau)_2 = (\delta_1 \after B\tau_1)_2 = \delta_2 T \after L\tau_2$ and
$(\tau \circledcirc \delta)_2 = (\tau_1 L \after T \delta_1)_2 =
\tau_2 B \after \delta_2$.  Moreover, it is easy to check that
$(\delta_1 \after B\tau_1 \after \kappa G)_2 = F \kappa
\after(\delta_1 \after B \tau_1)_2$.
The lower-right rectangle commutes since $((\Theta,a),\zeta)$ is a
$\overline{B}$-coalgebra.  The other two squares commute by
naturality.

For the second part of the theorem, let $c \colon X \rightarrow BT(X)$
be a coalgebra. Since $\emlog$ is an algebra morphism, the equation
$\emlog \after \trem_c = \trlog_c$ follows by uniqueness of coalgebra-to-algebra morphisms
from $c$ to $\calogbt$.  \QED
\end{proof}

The equality $\emlog \after \trem_c = \trlog_c$ means that equivalence
w.r.t.\ Eilenberg-Moore trace semantics implies equivalence w.r.t.\ the
logical trace semantics. The converse is, of course, true if $\emlog$
is monic. For that, it is sufficient if $\delta \colon BG \Rightarrow
GL$ is \emph{expressive}.  Here expressiveness is the property that
for any $B$-coalgebra, the unique coalgebra-to-algebra morphism to the
corecursive algebra $G_{\delta}(\Phi, \alpha^{-1})$ factors as a $B$-coalgebra
homomorphism followed by a mono. This holds in particular if the
components $\delta_{A} \colon BG(A) \rightarrow GL(A)$ are all monic
(in $\cat{C}$)~\cite{Klin07}.

\begin{lemma}
\label{lm:expr-final}
If $\delta \colon BG \Rightarrow GL$ is expressive, then $\emlog$ is
monic.  Moreover, if $\delta$ is an isomorphism, then $\emlog$ is an
iso as well.
\end{lemma}

\begin{proof}
Expressivity of $\delta$ means that we have $\emlog = m \after h$ for
some coalgebra homomorphism $h$ and mono $m$.  By finality of $\zeta$
there is a $B$-coalgebra morphism $h'$ such that $h' \after h =
\idmap$.  It follows that $h$ is monic (in $\Cat{C}$), so that $m
\after h = \emlog$ is monic too.

For the second claim, if $\delta$ is an isomorphism, then
$G(\alpha^{-1}) \after \delta \colon BG(\Phi) \rightarrow G(\Phi)$ is
an invertible corecursive $B$-algebra, which implies it is a final
coalgebra (see~\cite[Proposition 7]{CaprettaUustaluVene:recursive},
which states the dual).  It then follows from~\eqref{eq:theory-finalb}
that $\emlog$ is a coalgebra morphism from one final $B$-coalgebra to
another, which means it is an isomorphism.  \QED
\end{proof}

Previously, we have seen both a class of examples of the
Eilenberg-Moore approach (Theorem~\ref{thm:em-automata}), and the
logical approach (Proposition~\ref{prop:automaton-logic}).  Both arise
from the same data: a monad $T$ (just a functor in the logical
approach) and an $\EM$-algebra $t$.  We thus obtain, for these
automata-like examples, both a logical trace semantics and a matching
`Eilenberg-Moore' semantics, where the latter essentially amounts to a
determinisation procedure.  The underlying distributive laws
satisfy~\eqref{eq:kappa-rho} by construction, so that the two
approaches coincide (as already seen in the concrete examples).

\begin{theorem}
Let $\Omega$ be a set, $T \colon \Sets \rightarrow \Sets$ a monad and
$t \colon T(\Omega) \rightarrow \Omega$ an $\EM$-algebra.  The
$\EM$-law $\kappa$ of Theorem~\ref{thm:em-automata}, together with
$\delta, \tau$ as defined in the proof of
Proposition~\ref{prop:automaton-logic},
satisfies~\eqref{eq:kappa-rho}. For any coalgebra $c \colon X
\rightarrow \Omega \times T(X)^A$, the map $\trlog_c$ coincides (up to
isomorphism) with the map $\trem_c$.
\end{theorem}

\begin{proof}
To prove~\eqref{eq:kappa-rho}, i.e., $\delta \circledcirc
\tau \after \kappa = \tau \circledcirc \delta$, we first compute,
following~\eqref{eqn:logicextensions},
$$\begin{array}{rcl}
\lefteqn{(\delta \circledcirc \tau)_X \,\colon\,
   \Omega \times (T(\Omega^X))^A \longrightarrow \Omega^{A \times X + 1}}
\\
& = &
\delta_X \after (\idmap \times \tau_X^A)
\\
& = &
\delta_X \after (\idmap \times (t^X \after \st)^A)
\\
\lefteqn{(\tau \circledcirc \delta)_X \,\colon\,
    T(\Omega \times (\Omega^X)^A) \longrightarrow \Omega^{A \times X + 1}}
\\
& = &
\tau_{A \times X + 1} \after T(\delta_X)
\\
& = &
t^{A \times X + 1} \after \st \after T(\delta_X).
\end{array}$$

\noindent Hence, we need to show that 
\begin{equation}
\label{eq:toshow-rho}
\delta_X \after (\idmap \times (t^X \after \st)^A) \after (t \times \st) \after \langle T(\pi_1),T(\pi_2)\rangle 
		= t^{A \times X + 1} \after \st \after T(\delta_X)
	\end{equation}

\noindent for every set $X$. To this end, let $S \in T(\Omega \times
(\Omega^X)^A)$ and $u \in (A \times X + 1)$. We first spell out the
right-hand side:
$$\begin{array}{rcl}
&& (t^{A \times X + 1} \after \st \after T(\delta_X)(S))(u)\\
& = &
t((\st \after T(\delta_X)(S))(u)) 
\\
& = &
t(T(\evmap_u \after \delta_X)(S)) 
\\
& = &
\begin{cases} 
t(T(\pi_1)(S)) & \text{if } u=\ast \in 1 \\
t(T(\evmap_x \after \evmap_a \after \pi_2)(S)) & \text{if } u=(a,x) \in A \times X 
\end{cases}
\end{array}$$

\noindent In the last step, we used the definition of $\delta$:
	\begin{align*}
		\evmap_{\ast} \after \delta_X (\omega,f) &= \delta_X(\omega,f)(\ast) = \omega = \pi_1(\omega,f)\,, \\
		\evmap_{(a,x)} \after \delta_X (\omega, f) &= \delta_X(\omega,f)(a,x) = f(a)(x) = \evmap_x \after \evmap_a \after \pi_2(\omega,f) \,.
	\end{align*}
	For the left-hand side of~\eqref{eq:toshow-rho}, distinguish cases $\ast \in 1$ and $(a,x) \in A \times X$.
	\begin{align*}
		& (\delta_X \after (\idmap \times (t^X \after \st)^A) \after (t \times \st) \after \langle T(\pi_1),T(\pi_2)\rangle(S))(\ast) \\
		& = \pi_1(\idmap \times (t^X \after \st)^A) \after (t \times \st) \after \langle T(\pi_1),T(\pi_2)\rangle(S)) \\
		& = t(T(\pi_1)(S)) 
	\end{align*}
	which matches the right-hand side of~\eqref{eq:toshow-rho}. For $(a,x) \in A \times X$, we have:
	\begin{align*}
		& (\delta_X \after (\idmap \times (t^X \after \st)^A) \after (t \times \st) \after \langle T(\pi_1),T(\pi_2)\rangle(S))(a,x) \\
		& = (((t^X \after \st)^A \after \st)(T(\pi_2)(S)))(a)(x) \\
		& = (((t^X)^A \after \st^A \after \st)(T(\pi_2)(S)))(a)(x) \\
		& = ( t^X \after \st (\st(T(\pi_2)(S))(a)))(x) \\
		& = ( t^X \after \st (T(\evmap_a)(T(\pi_2)(S)))(x) \\
		& = ( t^X \after \st (T(\evmap_a \after \pi_2)(S)))(x) \\
		& = t (\st(T(\evmap_a \after \pi_2)(S))(x)) \\
		& = t (T(\evmap_x) \after T(\evmap_a \after \pi_2)(S)) \\
		& = t (T(\evmap_x \after \evmap_a \after \pi_2)(S)) 
	\end{align*}
	which also matches the right-hand side, hence we obtain~\eqref{eq:toshow-rho} as desired.
	
	Since~\eqref{eq:kappa-rho} is satisfied, it follows from Theorem~\ref{thm:logic-and-em} that
	$\emlog \after \trem_c = \trlog_c$. Since $\delta$ is an iso, 
	$\emlog$ is an iso as well by Lemma~\ref{lm:expr-final}. 
	\QED
\end{proof}

\subsection{Kleisli and Logic}
\label{sec:comp-kleisli-logic}

To compare the Kleisli approach to the logical approach, we combine
their assumptions. Further, similar to the comparison between
Eilenberg-Moore and logic in the previous section, we assume a first
compatibility criterion by requiring the components $\tau$ to be 
componentwise Eilenberg-Moore algebras. 
\begin{assumption}[Comparison Kleisli and Logic] In this subsection, we assume
an adjunction $F \dashv G$,
endofunctors $A,L$ and a monad $T$ as follows:
$$
\xymatrix{
\op{\cat{D}} \ar@/_2ex/[rr]_-G \save !L(.5) \ar@(dl,ul)^{L} \restore 
& \bot & \ar@/_2ex/[ll]_-{F} 
\cat{C}\ar@/^2ex/[rr]^-{J} \ar@(ul,ur)^{TA} & \bot & \Kl(T) \ar@/^2ex/[ll]^-V \save !R(.5) \ar@(ur,dr)^{\overline{A}} \restore 
}
$$

\noindent together with:
\begin{enumerate}
\item An initial algebra $\beta \colon A(\Psi) \congrightarrow \Psi$. 

\item A $\Kl$-law $\lambda \colon AT \Rightarrow TA$, or equivalently,
  an extension $\overline{A}\colon \Kl(T) \rightarrow \Kl(T)$
  of $A\colon\cat{C} \rightarrow \cat{C}$.

\item $(\Psi,J(\beta^{-1}))$ is a final $\overline{A}$-coalgebra. 

\item An initial algebra $\alpha \colon L(\Phi) \congrightarrow \Phi$.

\item A step $\delta \colon AG \Rightarrow GL$.

\item A step $\tau\colon TG \Rightarrow G$, whose components are
  $\EM$-algebras (a monad action).
\end{enumerate}
\end{assumption}

\noindent By the last assumption, $\tau$ satisfies the equivalent conditions in Lemma~\ref{lm:eq-tau};
again, this is in itself not part of the logical approach, but is used
in the comparison to the Kleisli approach to trace semantics. 
  We
obtain the following unique coalgebra-to-algebra morphism $\kllog$
from the initial $A$-algebra:
\begin{equation}\label{eq:theory-initB-second}
\vcenter{
\xymatrix@R-0.5pc@C=0.7cm{
\Psi  \ar[rr]^{\kllog} & & G(\Phi) 
\\
& & GL(\Phi) \ar[u]_{G(\alpha^{-1})} \\
A(\Psi)\ar[uu]^{\beta}_{\cong} \ar[rr]^{A(\kllog)} 
	& 
	& AG(\Phi) \ar[u]_{\delta} \\
}}
\end{equation}

\noindent Since $\tau$ is a monad action, for every $X$, $G(X)$
carries an Eilenberg-Moore algebra $\tau_{X}$.  Thus we can take the
adjoint transpose (w.r.t.\ the Eilenberg-Moore adjunction)
$\overline{\kllog} = \tau_{\Phi} \after T(\kllog) \colon T(\Psi)
\rightarrow G(\Phi)$. We then have the following analogue of
Theorem~\ref{thm:logic-and-em}.

\begin{lemma}
\label{lm:cond-kleisli-log}
The distributive law $\lambda\colon AT \Rightarrow TA$ commutes with
the logics in~\eqref{eqn:logicextensions}, as in:
\begin{equation}
\label{eq:lambda-rho}
	\vcenter{
	\xymatrix@R-0.5pc{
		ATG \ar[rr]^{\lambda G} \ar[dr]_-{\delta \circledcirc \tau}
			& & TAG \ar[dl]^-{\tau \circledcirc \delta} \\
			& GL & 
	}
	}
\end{equation}	

\noindent iff there is a natural transformation $\widehat{\delta}
\colon L\widehat{F} \Rightarrow \widehat{F}\overline{A}$ satisfying
$\widehat{\delta} J = \delta$ in:
\[ \xymatrix@R-0.5pc{
& & \Kl(T)\ar[dll]_{\widehat{F}}\save !R(.5) \ar@(ur,dr)^{\overline{A}} \restore 
  & \quad\mbox{with}\quad
  L\widehat{F} \ar@{=>}[r]^-{\widehat{\delta}} & \widehat{F}\overline{A}
\\
\op{\cat{D}}\save !L(.5) \ar@(dl,ul)^{L} \restore 
\\
& & \cat{C}\ar[ull]^{F}\ar[uu]_{J}\save !R(.5) \ar@(ur,dr)^{A} \restore 
  & \quad\mbox{with}\quad
  LF \ar@{=>}[r]^-{\delta} & FA
} \]

\noindent The two natural transformation on the right are written in
$\cat{D}$ instead of $\op{\cat{D}}$. The functor $\widehat{F} \colon
\Kl(T) \rightarrow \op{\cat{D}}$ is the extension corresponding to
$\tau$, Lemma~\ref{lm:eq-tau}.
\end{lemma}

\begin{proof}
The condition $\widehat{\delta} J = \delta$ simply means that
$\widehat{\delta}_X = \delta_X$ for every object $X$ in $\cat{C}$.
Naturality of $\widehat{\delta}$ amounts to commutativity of the
outside of the diagram below, for every map $f \colon X \rightarrow
T(Y)$.
	$$
	\xymatrix@R=0.4cm{
		& LF(Y) \ar[rr]^-{\delta} \ar[dd]_{L(\tau)}
                \ar`l[ld]`[dddr]_{L\widehat{F}(f)}[ddd]
			& & FA(Y) \ar[d]^{\tau}
                \ar`r[rd]`[dddl]^{\widehat{F}\overline{A}(f)}[ddd] &
                \\
                & & & FTA(Y) \ar[d]^{F(\lambda)} &
                \\
		& LFT(Y) \ar[rr]^-{\delta} \ar[d]_{LF(f)}
			& & FAT(Y) \ar[d]^{FA(f)} &
                \\
		& LF(X) \ar[rr]_-{\delta}
			& & FA(X) &
	}
	$$ 

\noindent The lower rectangle commutes by naturality, the upper is
equivalent to~\eqref{eq:lambda-rho}.  Hence,~\eqref{eq:lambda-rho}
implies naturality. Conversely, if $\widehat{\delta}$ is natural, then
the upper rectangle commutes for each $Y$ by taking $f = \idmap[TY]$
(the identity map in $\cat{C}$).  \QED
\end{proof}

\begin{theorem}
\label{thm:kl-log-comparison}
If the equivalent conditions in Lemma~\ref{lm:cond-kleisli-log} hold,
then the map $\overline{\kllog} = \tau_\Phi \after T(\kllog) \colon
T(\Psi) \rightarrow G(\Phi)$ is a $TA$-algebra morphism from $\cakl$ to
$\calogtb$, as on the left below.
	$$
	\xymatrix@R-0.5pc@C=1.8cm{
		TAT(\Psi) \ar[d]_-{\cakl} \ar[r]^-{TA(\smash{\overline{\kllog}})}
			& TAG(\Phi) \ar[d]^-{\calogtb} \\
		T(\Psi) \ar[r]^-{\smash{\overline{\kllog}}} 
			& G(\Phi)  
	}
	\qquad\qquad	
	\xymatrix@R-0.5pc{
		& X \ar[dl]_-{\trkl_c} \ar[dr]^(0.49){\trlog_c} & \\
		 T(\Psi) \ar[rr]^-{\smash{\overline{\kllog}}} & &  G(\Phi)
	}
	$$

\noindent In that case, for any coalgebra $c \colon X \rightarrow
TA(X)$ there is a commuting triangle as on the right above.
\end{theorem}

\begin{proof}
Consider the following diagram.
	$$
	\xymatrix@C=1.3cm@R=0.5cm{
		TAT(\Psi) \ar[d]^{T(\lambda)} \ar[r]^-{TAT(\kllog)} \ar`l[d]`[dddr]|{\cakl}[ddd]
			& TATG(\Phi) \ar[rr]^-{TA(\tau)} \ar[d]^{T(\lambda)}
			& & TAG(\Phi) \ar[d]_{T(\delta)} \ar`r[d]`[dddl]|{\calogtb}[ddd] \\
		TTA(\Psi) \ar[d]^{\mu } \ar[r]^{TTA(\kllog)} 
			& TTAG(\Phi) \ar[r]^{TT\delta} \ar[d]^{\mu}
			& TTGL(\Phi) \ar[r]^{T(\tau)} \ar[d]^{\mu}
			& TGL(\Phi) \ar[d]_{\tau}\\
		TA(\Psi) \ar[r]^{TA(\kllog)} \ar[d]^{T(\beta)}
			& TAG(\Phi) \ar[r]^{T(\delta)}
			& TGL(\Phi) \ar[r]^{\tau} \ar[d]^{TG(\alpha^{-1})}
			& GL(\Phi) \ar[d]_{G(\alpha^{-1})}\\
		T(\Psi) \ar[rr]_-{T(\kllog)}
			& & TG(\Phi) \ar[r]_-{\tau}
			& G(\Phi)
	}
	$$

\noindent Everything commutes: the upper right rectangle by
assumption~\eqref{eq:lambda-rho}, the right-most square in the middle
row since $\tau$ is an action, the outer shapes by definition of
$\cakl$ and $\calogtb$, the lower left rectangle by
\eqref{eq:theory-initB-second} and the rest by naturality.  

For the second part of the theorem, since $\overline{\kllog}$ is an algebra morphism,
we have that $\overline{\kllog} \circ \trkl_c$ is a coalgebra-to-algebra morphism
from $c$ to the corecursive algebra $\calogtb$. Hence $\overline{\kllog} \circ \trkl_c = \trlog_c$ 
by uniqueness of such morphisms. 
\QED
\end{proof}

The above result gives a sufficient condition under which Kleisli
trace equivalence implies logical trace equivalence.  However,
contrary to the case of traces in Eilenberg-Moore, in
Lemma~\ref{lm:expr-final}, we currently do not have a converse.  
%If
%$\delta$ has monic components, then it is easy to use corecursiveness
%to define a map from $\calogtb$ to $\cakl$, but this surprisingly is
The condition that $\delta$ has monic components is, surprisingly,
not sufficient for $\overline{\kllog}$ to be monic, as confirmed
by Example~\ref{ex:strange-rho} below. In the comparison between
Eilenberg-Moore and Kleisli traces in
Section~\ref{sec:comp-kleisli-em}, a similar difficulty arises.

\begin{example}\label{ex:strange-rho}
We give an example where $\delta \colon AG \Rightarrow GL$ is monic
and~\eqref{eq:lambda-rho} commutes, but where nevertheless logical
equivalence does not imply Kleisli trace equivalence.  Let
$\Cat{C} = \Cat{D} = \Sets$, $F = G = 2^{-}$, $A = L = (-) + 1$, 
$T = \Powfin$, $\tau \colon \Powfin 2^{-} \Rightarrow 2^{-}$ given
by union as before, and define the step $\delta$, for
$\varphi \in 2^X$, by $\delta_X(\varphi)(t) = \top$ iff 
$t \in X \wedge
\varphi(t)$, and $\delta_X(\ast)(t) = \top$ (the latter differs from
the step in Proposition~\ref{prop:tree}). Notice that $\delta$ indeed
has monic components.
(In the conference version~\cite{JacobsLR18}, we used a more general setting 
with $A = L = (\Sigma \times -) + 1$, where $\Sigma$ is a fixed set. However,
the associated $\delta$ is not monic 
if $\Sigma$ contains more than one element, contrary to what is stated there. Indeed,
we need to choose $\Sigma$ to be a singleton for the example to go through.)
	
Let $\lambda \colon AT \Rightarrow TA$ be the distributive law
from~\cite{HasuoJS07}, given by $\lambda_X(S) = \{x \mid x \in
S\}$ for $S \in \Powfin(X)$, and $\lambda(\ast) = \{\ast\}$. Then~\eqref{eq:lambda-rho} 
is satisfied:
	$$
	\xymatrix{
		\Powfin(2^X) + 1 \ar[rr]^{\lambda} \ar[d]_{\tau + 1}
			& & \Powfin(2^X + 1) \ar[d]^{\Powfin(\delta)} \\
		2^X + 1 \ar[r]^{\delta}
			& 2^{X + 1}
			& \Powfin(2^{X + 1}) \ar[l]_{\tau}
	}
	$$

\noindent It is straightforward to check that this commutes.  However,
given a coalgebra $c \colon X \rightarrow \Powfin A(X)$, the induced logical
semantics $\trlog \colon X \rightarrow 2^{\mathbb{N}}$ is: $\trlog(x)(n)
= \top$ iff $\ast \in c(x)$ or $n > 0 \wedge \exists y
\in c(x). \, \trlog(y)(n-1)=\top$.  In
particular, this means that if $\ast \in c(x)$ and $\ast \in c(y)$ for
some states $x,y$, then they are trace equivalent. This differs from
the Kleisli semantics, which amounts to the usual language semantics
of non-deterministic automata (over a singleton alphabet)~\cite{HasuoJS07}.
\end{example}

C{\^{\i}}rstea~\cite{Cirstea16} compares logical traces to a
`path-based semantics', which resembles the Kleisli approach (as well
as~\cite{KurzMPS15}) but does not require a final
$\overline{A}$-coalgebra.  In particular, given a commutative monad
$T$ on $\Sets$ and a signature $\Sigma$, she considers a canonical
distributive law $\lambda \colon H_\Sigma T \Rightarrow T H_\Sigma$,
which coincides with the one in~\cite{HasuoJS07}. C{\^{\i}}rstea shows
that, with $\Omega = T(1)$, $t = \mu_1 \colon TT(1) \rightarrow T(1)$
and $\delta$ from the proof of Proposition~\ref{prop:tree} (assuming
$T1$ to have enough structure to define that logic), the
triangle~\eqref{eq:lambda-rho} commutes (see~\cite[Lemma
  5.12]{Cirstea16}).

\subsection{Kleisli and Eilenberg-Moore}
\label{sec:comp-kleisli-em}

To compare the Eilenberg-Moore and Kleisli approaches, 
we first combine their assumptions. The Kleisli approach applies to $TA$-coalgebras; 
to match this, we make use of the variant of the Eilenberg-Moore approach for $TA$-coalgebras
presented in Section~\ref{sec:gen-em}.
The latter approach uses a lifting of a functor $B$
as well as a step relating $A$ and $B$. 

\begin{assumption}[Comparison Kleisli and Eilenberg-Moore] In this subsection, 
we assume to endofunctors $A,B$ and a monad $T$, on a base category $\cat{C}$,
and liftings $\overline{A}, \overline{B}$ to Kleisli- and
Eilenberg-Moore-categories as follows:
$$
\xymatrix{
\EM(T) \ar@/_2ex/[rr]_-U \save !L(.5) \ar@(dl,ul)^{\overline{B}} \restore 
& \bot & \ar@/_2ex/[ll]_-{\free} 
\;\cat{C}\;\ar@/^2ex/[rr]^-{J} \ar@(ul,ur)^{TA} & \bot & \Kl(T) \ar@/^2ex/[ll]^-V \save !R(.5) \ar@(ur,dr)^{\overline{A}} \restore
}
$$

\noindent In this situation we further assume the following ingredients,
which combine earlier assumptions.
\begin{enumerate}
\item An initial algebra $\beta \colon A(\Psi) \congrightarrow \Psi$. 

\item A $\Kl$-law $\lambda \colon AT \Rightarrow TA$, or equivalently,
  an extension $\overline{A}\colon \Kl(T) \rightarrow \Kl(T)$ of the
  functor $A$.

\item $(\Psi,J(\beta^{-1}))$ is a final $\overline{A}$-coalgebra.

\item An $\EM$-law $\kappa \colon TB \Rightarrow BT$, or equivalently,
  a lifting $\overline{B} \colon \EM(T) \rightarrow \EM(T)$.

\item A final coalgebra $\zeta \colon \Theta \congrightarrow
  B(\Theta)$.

\item A step $\rho \colon AU \Rightarrow U\overline{B}$.
\end{enumerate}
\end{assumption}

The step $\rho \colon AU \Rightarrow U\overline{B}$ is an assumption of 
the Eilenberg-Moore approach for $TA$-coalgebras in Section~\ref{sec:gen-em}, 
defined on top of the assumptions
for the Eilenberg-Moore approach for $BT$-coalgebras.
Under a further assumption such a law corresponds to an \emph{extension}
natural transformation as in~\cite{JacobsSS15}, see Proposition~\ref{prop:cond-kleisli-em}. 

Recall from Section~\ref{sec:em} that the final
$B$-coalgebra $(\Theta,\zeta)$ gives rise to a final
$\overline{B}$-coalgebra $((\Theta,a),\zeta)$.  We will make use of
the counit $\varepsilon$ of the $\EM$-adjunction $\free \dashv U$ as a
step $U\varepsilon \colon TU \Rightarrow U$. Its components are
$\EM$-algebras. For the trace semantics
of $TA$-coalgebras via Eilenberg-Moore, see Section~\ref{sec:gen-em},
we make use of the composed step:
\begin{equation}
\begin{array}{rcl}
U\varepsilon \circledcirc \rho
& = &
\xymatrix{\Big(TAU\ar@{=>}[r]^-{T\rho} &
 TU  \overline{B}  \ar@{=>}[r]^-{U\varepsilon  \overline{B} } & U  \overline{B}  \Big) }
\end{array}.
\end{equation}

\noindent These assumptions form an instance of the assumptions in
Section~\ref{sec:comp-kleisli-logic}, where we compared Kleisli to
logical trace semantics.  In particular, in the latter we instantiate
$\op{\Cat{D}}$ with $\EM(T)$, $L$ with $\overline{B}$, $\delta$ with
$\rho\colon AU \Rightarrow U\overline{B}$ and $\tau$ with
$U\varepsilon\colon UT\Rightarrow U$. Thus, we immediately obtain the
comparison result from Theorem~\ref{thm:kl-log-comparison}.  For
presentation purposes, we restate the relevant results and
definitions.

There is the following unique coalgebra-to-algebra morphism $\klem$ from
the initial $A$-algebra:
%% \begin{equation}
%% \label{eq:theory-initB}
%% \vcenter{
%% \xymatrix@R-0.5pc{
%% \Psi \ar[d]_{\beta^{-1}} \ar[rrr]^-{\klem}
%% 	& & & \Theta \\
%% A(\Psi) \ar[r]^-{A(\klem)} 
%% 	& AU(\Theta,a) \ar[r]^-{\rho}
%% 	& U\overline{B}(\Theta,a) \ar@{=}[r]
%% 	& B(\Theta) \ar[u]_{\zeta^{-1}} 
%% }}
%% \end{equation}

\begin{equation}
\label{eq:theory-initB}
\vcenter{\xymatrix@R-1.2pc{
\Psi\ar@{-->}[rr]^-{\klem} & & \Theta
\\
& & B(\Theta)\ar[u]_{\zeta^{-1}}
\\
& & U\overline{B}(\Theta,a)\ar@{=}[u] 
   & BT(\Theta)%\ar`u[ul]^-{}`l_{} [ul]^-{B(a)} &
   \ar@/_1em/[ul]_-{B(a)}\ar`u[uul]_-{\caem}`l_{} [uul] &
\\
& & AU(\Theta,a)\ar[u]_-{\rho} & TA(\Theta)\ar[u]_-{U(\rho_{2})}
\\
A(\Psi)\ar@{-->}[rr]^-{A(\klem)}\ar[uuuu]^{\beta}_{\cong} 
   & & A(\Theta)\ar@{=}[u]\ar`r[ru]^-{\eta}`u_{} [ru] &
}}
\end{equation}

\noindent The rectangle on the right commutes since $\caem = \zeta^{-1}
\after B(a)$, by definition~\eqref{eqn:caem}, and:
\[ \begin{array}{rcll}
B(a) \after U(\rho_{2}) \after \eta
%& = &
%B(a) \after B(\mu) \after \kappa \after T(\rho) \after TA(\eta) \after \eta
%   \quad \mbox{by Theorem~\ref{thm:mates}} 
%\\
%& = &
%B(a) \after B(\mu) \after \kappa \after \eta \after \rho \after A(\eta)
%\\
%& = &
%B(a) \after B(\mu) \after B(\eta) \after \rho \after A(\eta)
%\\
& = &
B(a) \after \rho \after A(\eta) \qquad \qquad &\text{by Lemma~\ref{lm:mates-useful}}
\\
& = &
\rho \after A(a) \after A(\eta) & \text{since $a \colon 
   (T(\Theta),\mu) \rightarrow (\Theta,a)$ in $\EM(T)$}
\\
& = &
\rho.
\end{array} \]

Taking the adjoint transpose, w.r.t.\ the Eilenberg-Moore adjunction
$\free \dashv U$, of this map $\klem \colon \Psi \rightarrow \Theta =
U(\Theta, a)$, yields a map of Eilenberg-Moore algebras:
$$\begin{array}{rcccl}
\overline{\klem}
& = &
\xymatrix@C-0.5pc{
\Big(\free (\Psi) \ar[r]^-{\free(\klem)}
	& \free U(\Theta, a) \ar[r]^-\varepsilon
	& (\Theta,a)\Big)
}
& = &
\xymatrix@C-0.5pc{
\Big(T(\Psi) \ar[r]^-{T(\klem)}
	& TU(\Theta, a) \ar[r]^-{a} & \Theta\Big).
}
\end{array}$$

We have seen in~\eqref{eqn:cakl} that $T(\Psi)$ is the carrier of the
corecursive algebra $\cakl \colon TAT(\Psi) \rightarrow T(\Psi)$
giving Kleisli trace semantics. At the same time $\Theta$ is the
carrier of the corecursive algebra
$\caem \colon BT(\Theta)
\rightarrow \Theta$
from~\eqref{eqn:caem} as well as the corecursive algebra 
$$
G_{U\varepsilon \circledcirc \rho}((\Theta, a), \zeta) = \caem^A = 
\caem \after U(\rho_2) \colon
TA(\Theta) \rightarrow \Theta \,,
$$
where $((\Theta, a), \zeta)$ is the final $\overline{B}$-coalgebra,
and the equality on the right is given by Lemma~\ref{lm:corec-alg-ta}. 
Thus, the map $\overline{\klem} \colon
T(\Psi) \rightarrow \Theta$ relates the carriers of
the corecursive $TA$-algebras $\cakl$ and $\caem^A$.  Like in the previous sections, we
now give a sufficient condition for $\overline{\klem}$ to be an algebra morphism.

\begin{proposition}
\label{prop:cond-kleisli-em}
In the above setting, the following three statements are equivalent.
\begin{enumerate}
\item The distributive law $\lambda\colon AT \Rightarrow TA$ commutes
  with the two composed steps $\rho \circledcirc U(\varepsilon)$ and
  $U(\varepsilon) \circledcirc \rho$, as in:
\begin{equation}
\label{eq:lambda-rho-em}
	\vcenter{
	\xymatrix@R-0.5pc{
		ATU \ar[rr]^{\lambda} \ar[dr]_-{\rho \circledcirc U(\varepsilon)}
			& & TAU \ar[dl]^-{U(\varepsilon) \circledcirc \rho} \\
			& U\overline{B} & 
	}
	}
	\end{equation}	

\item There is a natural transformation $\mathfrak{e} \colon
  \widehat{\free} \overline{A} \Rightarrow \overline{B}
  \widehat{\free}$ satisfying $\mathfrak{e} J = \rho_2$ in:
\[ \xymatrix@R-0.5pc{
& & \Kl(T)\ar[dll]_{\widehat{\free} = K}\save !R(.5) \ar@(ur,dr)^{\overline{A}} \restore 
  & \quad\mbox{with}\quad
   \widehat{F}\overline{A} \ar@{=>}[r]^-{\mathfrak{e}} & \overline{B}\widehat{F}
\\
\EM(T)\save !L(.5) \ar@(dl,ul)^{\overline{B}} \restore 
\\
& & \cat{C}\ar[ull]^{\free}\ar[uu]_{J}\save !R(.5) \ar@(ur,dr)^{A} \restore 
  & \quad\mbox{with}\quad
  \free A  \ar@{=>}[r]^-{\rho_2} & \overline{B}\free
} \]

\noindent The functor $\widehat{\free} \colon \Kl(T) \rightarrow
\EM(T)$ is the extension corresponding to $U(\varepsilon)$, according
to Theorem~\ref{thm:mates}; it is often called the `comparison'
functor, and then written as $K$.

\item The following `extension requirement' from~\cite{JacobsSS15}
  commutes:
\begin{equation}
\label{eq:ext-req}
\vcenter{\xymatrix@R-0.5pc{
		TAT \ar[rr]^-{U(\rho_2)} \ar[d]_{T(\lambda)}
		& & BTT \ar[d]^{B(\mu)} \\
		TTA \ar[r]_{\mu}
			& TA \ar[r]_{U(\rho_2)}
			& BT
}}
\end{equation}
\end{enumerate}
\end{proposition}

\begin{proof}
The equivalence of points~(1) and~(2) is an instance of
Lemma~\ref{lm:cond-kleisli-log}, where it should be noted that we are
instantiating $\cat{D}$ with $\op{\EM(T)}$. 
which causes the two
natural transformations in the above diagram to be in opposite
direction.

We show the equivalence of~(1) and~(3). Using
Lemma~\ref{lm:mates-composition}, it is straightforward to check, via
Theorem~\ref{thm:mates}, that $\big(\rho \circledcirc
U(\varepsilon)\big)_{2} = \big(\rho_1 \after AU\varepsilon\big)_2 =
\overline{B} \varepsilon \after \rho_2$ and $U\varepsilon \after
T\rho_1 \after \lambda = \rho_2 \after \varepsilon \after
\free(\lambda)$.  As a consequence, commutativity of
Diagram~\eqref{eq:lambda-rho-em} is equivalent to commutativity of the
following diagram:
$$\xymatrix@R-0.5pc{
		\free AT \ar[rr]^-{\rho_2} \ar[d]_{\free(\lambda)}
		& & \overline{B} \free T \ar[d]^{\overline{B}(\varepsilon)} \\
		\free TA \ar[r]_{\varepsilon}
			& \free A \ar[r]_{\rho_2}
			& U \overline{B} 			
}$$

\noindent This amounts to the diagram in point~(3). \QED
\end{proof}

Under the above equivalent conditions, we obtain the desired algebra morphism.

\begin{theorem}
\label{thm:kl-em-comparison}
If the equivalent conditions in Proposition~\ref{prop:cond-kleisli-em}
hold, then the map $\overline{\klem}\colon T(\Psi) \rightarrow \Theta$
obtained from~\eqref{eq:theory-initB}, is a $TA$-algebra morphism
between corecursive algebras $\cakl$ and $\caem^A =
\caem \after U(\rho_2)$, as on the left below.
	$$
	\xymatrix@R-0.5pc@C=1.3cm{
		TAT(\Psi) \ar[d]_-{\cakl} \ar[r]^-{TA(\smash{\overline{\klem}})}
			& TA(\Theta) \ar[d]^-{\caem^A = \caem \after U(\rho_2)} \\
		T(\Psi) \ar[r]^-{\smash{\overline{\klem}}} 
			& \Theta  
	}
	\qquad\qquad	
	\xymatrix@R-0.5pc{
		& X \ar[dl]_-{\trkl_c} \ar[dr]^(0.49){\trem_c^A} & \\
		 \llap{$T($}\Psi) \ar[rr]^-{\smash{\overline{\klem}}} & &  \Theta
	}
	$$

\noindent In that case, for any coalgebra $c \colon X \rightarrow
TA(X)$ there is a commuting triangle as on the right above, where
$\trem_c^A$ is the unique map from $(X,c)$ to the corecursive algebra
$\caem \after U(\rho_2)$.
\end{theorem}

\begin{proof}
In order to prove commutation of the rectangle we need to combine
many earlier facts:
\[ \begin{array}[b]{rcll}
\overline{\klem} \after \cakl
& \smash{\stackrel{\eqref{eqn:cakl}}{=}} &
a \after T(\klem) \after T(\beta) \after \mu \after T(\lambda)
\\
& \smash{\stackrel{\eqref{eq:theory-initB}}{=}} &
a \after T(\zeta^{-1} \after B(a) \after U(\rho_{2}) \after \eta \after A(\klem))
   \after \mu \after T(\lambda)
\\
& \smash{\stackrel{\eqref{eq:final-coalg-em}}{=}} &
\zeta^{-1} \after B(a) \after \kappa \after TB(a) \after TU(\rho_{2}) 
   \after T(\eta) \after TA(\klem) \after \mu \after T(\lambda)
\\
& = &
\zeta^{-1} \after B(a) \after BT(a) \after \kappa \after TU(\rho_{2}) 
   \after T(\eta) \after TA(\klem) \after \mu \after T(\lambda)
\\
& = &
\zeta^{-1} \after B(a) \after B(\mu) \after \kappa \after TU(\rho_{2}) 
   \after T(\eta) \after TA(\klem) \after \mu \after T(\lambda)
\\
& \smash{\stackrel{\eqref{eqn:extensiondiagsingle}}{=}} &
\zeta^{-1} \after B(a) \after U(\rho_{2}) \after \mu
   \after T(\eta) \after TA(\klem) \after \mu \after T(\lambda)
\\
& = &
\zeta^{-1} \after B(a) \after U(\rho_{2})
     \after TA(\klem) \after \mu \after T(\lambda)
\\
& = &
\zeta^{-1} \after B(a) \after BT(\klem) \after U(\rho_{2}) \after 
   \mu \after T(\lambda)
\\
& \smash{\stackrel{\eqref{eq:ext-req}}{=}} &
\zeta^{-1} \after B(a) \after BT(\klem) \after B(\mu) \after U(\rho_{2})
\\
& = &
\zeta^{-1} \after B(a) \after B(\mu) \after BT^{2}(\klem) \after U(\rho_{2})
\\
& = &
\zeta^{-1} \after B(a) \after BT(a) \after BT^{2}(\klem) \after U(\rho_{2})
\\
& = &
\caem \after BT(\overline{k}) \after U(\rho_{2})
\\
& = &
\caem \after U(\rho_{2}) \after TA(\overline{k}).
\end{array} \]

Now let a coalgebra $c\colon X \rightarrow TA(X)$ be given.
We need to prove that $\overline{\klem} \after \trkl_{c}$ satisfies
the defining property of $\trem_c^A$.  But this easy using
the rectangle in the theorem:
\[ \begin{array}[b]{rcccl}
\caem^{A} \after TA(\overline{\klem} \after \trkl_{c}) \after c
& = &
\overline{\klem} \after \cakl \after TA(\trkl_{c}) \after c
& \smash{\stackrel{\eqref{diag:trkl}}{=}} &
\overline{\klem} \after \trkl_{c}
\end{array} \eqno{\QEDbox} \]
\end{proof}

Just like in the comparison between Kleisli and logic, the above
result gives a sufficient condition for the Eilenberg-Moore trace
semantics to factor through the Kleisli trace semantics. However,
again we do not know of a reasonable condition to ensure that the map
$\overline{\klem}$ is monic. Such a result is important for the
comparison: it would ensure that two states are equivalent
w.r.t.\ Kleisli traces iff they are equivalent w.r.t.\ Eilenberg-Moore
traces (right now, we only have the implication from left to right).
In~\cite{JacobsSS15}, such a condition is also missing; monicity of
$\overline{\klem}$ is only shown to hold in several concrete examples.

In~\cite[\S6]{JacobsSS15}, the Kleisli approach to coalgebraic trace
semantics is compared with the Eilenberg-Moore approach, making use of
an `extension' natural transformation $\mathfrak{e}$ satisfying two 
requirements, namely:
\[ \vcenter{\xymatrix@R-0.5pc{
TAT\ar[d]_{\mathfrak{e}}\ar[r]^-{T(\lambda)} & T^{2}A\ar[r]^-{\mu} & TA\ar[d]^{\mathfrak{e}}
& 
T^{2}A\ar[rr]^-{\mu}\ar[d]_{T(\mathfrak{e})} & & TA\ar[d]^{\mathfrak{e}}
\\
BT^{2}\ar[rr]^-{B(\mu)} & & BT
& 
TBT\ar[r]^-{\kappa} & BT^{2}\ar[r]^-{B(\mu)} & BT
}} \]

\noindent 
In the present step-based setting, 
the rectangle on the right
occurred as~\eqref{eqn:extensiondiagsingle} in Section~\ref{sec:gen-em},
which contains the generalisation of
Eilenberg-Moore trace semantics
that is used here. The first of the
above two rectangles captures compatibility in
Proposition~\ref{prop:cond-kleisli-em} and is used for a comparison of
Kleisli and Eilenberg-Moore semantics in
Theorem~\ref{thm:kl-em-comparison}. The conclusion is that the
approach of this paper not only covers the approach
of~\cite[\S6]{JacobsSS15} but also puts it in a wider step-based
perspective, using corecursive algebras.

\section{Completely Iterative Algebras}\label{sec:cia}

%\marginpar{Do we wish to keep this section?}  Yes please, says Paul.

Milius~\cite{Milius:cia} introduced a notion of ``complete iterativity'' of algebras, that is stronger than corecursiveness and has the advantage of being preserved by various constructions.  So, whenever we encounter a corecursive algebra, it is natural to ask whether it is in fact completely iterative.  This section shows that all our corecursive algebras are completely iterative (Theorem~\ref{thm:finalcia}), and that this yields trace maps in more general settings.
\begin{definition}
Let $\catc$ have binary coproducts.  For an endofunctor $H$ on $\catc$, an $H$-algebra $a \colon HA \to A$
is \emph{completely iterative} when $[\idmap,a]$ is a corecursive
$A + H$-coalgebra.  Explicitly: when for every $c \colon X \to A+ HX$
there is a unique $f \colon X \to A$ such that the following diagram
commutes.
\[     \xymatrix@R-0.5pc{
 X \ar[rr]^{f} \ar[d]_{c} & & A  \\
A+HX \ar[rr]_{A+Hf} & & A+HA \ar[u]_{[\idmap,a]}
} \]
\end{definition}

The following gives two useful ways of constructing such algebras.
\begin{proposition}\hfill
  \begin{enumerate}
  \item \label{item:fincia} If $\zeta \colon A \to HA$ is a final $H$-coalgebra, then $(A,\zeta^{-1})$ is completely iterative.
 \item \label{item:grhocia} Given a step as in Section~\ref{sect:stepsem}, the functor $G_{\rho}$ preserves complete iterativity. 
  \end{enumerate}
\end{proposition}
\begin{proof}
  Part~(\ref{item:fincia}) is included in~\cite[Theorem 2.8]{Milius:cia}, and part~(\ref{item:grhocia}) is the dual of~\cite[Theorem 5.6]{HinzeWuGibbons:conjhylo}. \QED
\end{proof}
We may thus say: \emph{``step-induced algebra liftings of right adjoints preserve complete iterativity''}.  We deduce the following strengthening of Theorem~\ref{thm:finalcor}.
\begin{theorem} \label{thm:finalcia}
  Given a step as in Section~\ref{sect:stepsem} and a final coalgebra  $\zeta \colon A \to HA$, the algebra $G_{\rho}(A,\zeta^{-1})$ is completely iterative.
\end{theorem}

if $L$ has a final coalgebra $(\Psi,\zeta)$ then $G_{\rho}(A,\zeta^{-1})$ is completely iterative.
% \begin{theorem}
%   	Suppose that $L$ has a final coalgebra $(\Psi, \zeta)$. Then for every
% 	$c \colon X \to A + HX$ there is a unique map $c^\dagger$ in
% 	$$
% 	\xymatrix{
% 		X \ar@{-->}[rr]^{c^\dagger} \ar[d]_{c} 
% 			& & G(\Psi)  \\
% 		A+H(X) \ar@{-->}[rr]^{[\idmap,H(c^\dagger)]}
% 			& & A+HG(\Psi) \ar[u]_{\overline{G}(\zeta^{-1})}
% 	}
% 	$$
% \end{theorem}

For example, in the setting of Section~\ref{sec:em}, we obtain the following variation of~(\ref{eq:traces-em}).  Given a coalgebra $c \colon X \to \Theta + BT(X)$, there is a unique map, writting as $\trem_c$, making the following square commute:
\begin{equation}
\label{eq:gen-traces-em}
\vcenter{\xymatrix@R-0.5pc{
	X \ar@{-->}[rr]^{\trem_c} \ar[d]_c
		& & \Theta \\
	\Theta + BT(X) \ar@{-->}[rr]_{\Theta + BT(\trem_c)}
        & & \Theta + BT(\Theta) \ar[u]_{[\idmap,\caem]}
}}
\end{equation}
In what sense is $\trem_c$ a ``trace map''?  Let us look at the special case of Example~\ref{ex:first-em}, where $B(X) = 2 \times X^A$ so that $\Theta$ is the set $2^{A^*}$ of languages, and $T$ is the
finite powerset monad.  Think of $c$ as a \emph{generalised} nondeterministic automaton: as well as the usual accepting and rejecting states, there can also be ``semantic states'' that are labelled with a language and from which there are no transitions.  For a state $x \in X$, a \emph{trace} of $x$ is either a word appearing along a path from $x$ to an accepting state (as usual), or a concatenation of words $s$ and $t$, where $s$ appears along a path from $x$ to a semantic state labelled by $L$, and $t \in L$.  With this definition, we see that $\trem_c$ sends $x \in X$ to its set of traces.

Each of our examples is similar to this one: the completely iterative algebra yields trace semantics for a generalised transition system in which the semantics may sometimes be given directly.

\section{Future work}
\label{sec:fw}

The main contribution of this paper is a general treatment of trace
semantics via corecursive algebras, constructed through an adjunction
and a step, covering the `Eilenberg-Moore', `Kleisli' and `logic'
approaches to trace semantics. It is expected that our framework also
works for other examples, such as the `quasi-liftings'
in~\cite{Bonchi0S17}, but this is left for future work.
In~\cite{KerstanKoenigWesterbaan:liftadj}, several examples of
adjunctions are discussed in the context of automata theory, some of
them the same as the adjunctions here, but with the aim of lifting
them to categories of coalgebras, under the condition that what 
we call the step is an iso. In our case, it usually is not an iso, since the behaviour
functor is a composite $TB$ or $BT$; however, it remains interesting
to study cases in which such adjunction liftings appear, as used for
instance in the aforementioned paper and~\cite{Rot16,KlinR16}.  Further,
our treatment in Section~\ref{sec:em} (Eilenberg-Moore) assumes a
monad to construct the corecursive algebra, but it was shown by
Bartels~\cite{Bartels03} that this algebra is also corecursive when
the underlying category has countable coproducts (and dropping the
monad assumption). We currently do not know whether this fits our
abstract approach.  
Finally, the Kleisli/logic and Kleisli/Eilenberg-Moore comparisons (Section~\ref{sec:comp}) are 
similar, but the Eilenberg-Moore/logic
comparison seems different. So far we have 
been unable to derive a general perspective on such comparisons that
covers all three.

A further direction of research is provided by the recent~\cite{DorschMS19}, where graded monads are 
used to define trace semantics in a general way, together with associated expressive logics. On the one 
hand it would be interesting to try and capture this within our steps-and-adjunctions framework; but this 
would require to capture graded semantics via some notion of finality or corecursiveness, which we are 
not currently aware of. On the other hand, as pointed out by one of the reviewers, the comparison results 
of Section~\ref{sec:comp} can be viewed as expressivity results of one semantics w.r.t. the other---this 
insight is interesting on its own, and might help in devising a more general method for making comparisons 
as in Section~\ref{sec:comp}, also discussed above. Further, the expressiveness criteria in~\cite{DorschMS19} 
may be useful to address the issues in the comparison of Kleisli semantics to Eilenberg-Moore and logical 
trace semantics. We leave these considerations for future work. 

\paragraph{Acknowledgement.} We are grateful to the anonymous referees of both the conference version and
this extended paper for various comments and suggestions.

\bibliography{logic-computation.bib}

\newpage
\begin{appendix}

\section{Details for Section~\ref{sect:stepsem}}\label{app:details-steps}

We recall a (standard) lemma that relates a step
$\rho_1$ to its mate $\rho_2$.

\begin{lemma}
\label{lm:mates-useful}
For any step $\rho \colon HG \Rightarrow GL$, the following diagrams
commute.
	$$
	\xymatrix{
		FHG \ar[r]^{F\rho_1} \ar[d]_{\rho_2 G}
			& FGL \ar[d]^{\varepsilon L}\\
		LFG \ar[r]_{L \varepsilon} 
			& L
	}
	\qquad\qquad
	\xymatrix{
		H \ar[r]^{H\eta} \ar[d]_{\eta H}
			& HGF \ar[d]^{\rho_1 F}\\
		GFH \ar[r]_{G\rho_2} 
			& GLF
	}
	$$
\end{lemma}

\begin{proof}
By unpacking the definitions and using naturality, e.g.,~in the first diagram:
\[ \begin{array}[b]{rcl}
L\varepsilon \after \rho_{2}G
& = &
L\varepsilon \after \varepsilon LFG \after F\rho_{1} FG \after FH\eta G
\\
& = &
\varepsilon L \after FGL\varepsilon \after F\rho_{1} FG \after FH\eta G
\\
& = &
\varepsilon L \after F\rho_{1} \after FHG\varepsilon \after FH\eta G
\\
& = &
\varepsilon L \after F\rho_{1}.
\end{array} \eqno{\QEDbox} \]
\end{proof}

The above lemma is useful in the proofs of Theorem~\ref{thm:mates} 
and Lemma~\ref{lm:mates-composition}, presented next.

\begin{proof}[Proof of Theorem~\ref{thm:mates}]
	The correspondence between $\rho_1, \rho_2, \rho_3$ and $\rho_4$ follows from
	Theorem~\ref{thm:matesgeneral}. 
	For the second part of the statement, 
	suppose $H$ and $L$ have monad structures $(H,\eta^H,\mu^H)$ and $(L,\eta^L,\mu^L)$ respectively. 
	The fact that $\rho_1$ is an $\EM$-law iff 
	$\rho_4$ is a monad map can be reconstructed from~\cite{street}.
	
	We show that if $\rho_1$ is an $\EM$-law then $\rho_2$ is a Kleisli law---the
	converse follows analogously. To this end, for the compatibility of 
	$\rho_2$ with $\mu$, consider the following diagram.
	$$
	\xymatrix@C=1.5cm{
		FHH \ar[rr]^{F\mu^H} \ar[d]^{\rho_2 H} \ar[dr]^{FHH\eta}
			& 
			& FH \ar[d]_{FH\eta} \ar`r[d]`[ddddl]^{\rho_2}[dddd] \\ 
		LFH \ar[d]^{LFH\eta} \ar`l[d]`[dddr]_{L\rho_2}[ddd]
			& FHHGF \ar[dl]^{\rho_2 HGF} \ar[d]^{FH \rho_1 F} \ar[r]^{F \mu^H LF}
			& FHGF \ar[dd]_{F\rho_1 F} \\
		LFHGF \ar[d]^{LF\rho_1 F} 
			& FHGLF \ar[d]^{F\rho_1 LF} \ar[dl]^{\rho_2 GLF}
			& \\
		LFGLF \ar[d]^{L \varepsilon LF} 
			& FGLLF \ar[dl]^{\varepsilon LLF} \ar[r]^{FG \mu^L F}
			& FGLF \ar[d]_{\varepsilon LF} \\
		LLF \ar[rr]^{\mu^L F} 
			& 
			& LF 
	}
	$$
	The outside shapes commute by definition of $\rho_2$. For the (inner) rectangle, 
	everything commutes, clockwise starting at the north by naturality, 
	the fact that $\rho_1$ is an $\EM$-law, naturality, Lemma~\ref{lm:mates-useful},
	and twice naturality. 
	For the unit axiom, we have the following diagram:
	$$
	\xymatrix@C=1.5cm{
		FH \ar[r]_{FH \eta} \ar`u`[rrr]^{\rho_2}[rrr]
			& FHGF \ar[r]_{F\rho_1 F} 
			& FGLF \ar[r]_{\varepsilon LF}
			& LF \\
		F \ar[u]^{F \eta^H} \ar[r]^{F \eta} \ar`d[r]`[rrru]^{\id}[rrr]
			& FGF \ar@{=}[r] \ar[u]_{F\eta^H GF}
			& FGF \ar[r]^{\varepsilon F} \ar[u]_{FG \eta^L F}
			& F \ar[u]_{\eta^L F}
	}
	$$
	which commutes by naturality, a triangle identity of the adjunction, definition of $\rho_2$ 
	and the fact that $\rho_1$ is an $\EM$-law (middle shape). 
	
	Finally, the correspondence between $\EM$-laws and liftings was shown in~\cite{Johnstone:Adj-lif}, 
	and the variant for Kleisli laws in~\cite{Mulry93}. 
	\qed
\end{proof}

\begin{proof}[Proof of Lemma~\ref{lm:mates-composition}]
First, note that $\big(\theta \circledcirc \rho\big)_2 = \varepsilon
MLF \after F\theta_1 LF \after FK\rho_1 F \after FKH\eta$.  It thus
suffices to prove that the following diagram commutes.
$$\xymatrix@C=1.3cm{
FKH \ar[rr]^-{FKH \eta} \ar[d]_-{\theta_2 H}		
	& & FKHGF \ar[d]^-{FK\rho_1 F} \ar[dl]_-{\theta_2 HGF} 
\\
MFH \ar[dd]_-{M\rho_2} \ar[r]^-{MFH\eta}
	& MFHGF \ar[d]_-{MF\rho_1 F}
	& FKGLF \ar[dl]_-{\theta_2 GLF} \ar[dd]^-{F\theta_1 LF} 
\\
& MFGLF \ar[dl]_-{M\varepsilon LF} & 
\\
MLF & & FGMLF \ar[ll]^-{\varepsilon MLF}
}$$

\noindent All the inner parts commute, clockwise starting from the top
by naturality of $\theta_2$ (twice), Lemma~\ref{lm:mates-useful}, and
definition of $\rho_2$ from $\rho_1$.  \QED
\end{proof}

\section{Steps and bimodules} \label{sec:bimod}

% Typo fixpoint of C(X,A) should be C(X,Phi)

For the example of partial traces for I/O given in Section~\ref{sec:partial}, it is convenient to take a different view of our step-and-adjunction setting, using the following notion.
\begin{definition}
  For categories $\cat{C}$ and $\cat{D}$, a \emph{bimodule}
    $\cat{O} \colon \cat{C} \relto \cat{D}$ consists of the following data.
 \begin{itemize}
 \item A family of sets $(\cato(X,\uly))_{X \in \cat{C}, \uly \in \cat{D}}$, where $g \in \cato(X,\uly)$ is called an \emph{$\cat{O}$-morphism} $g \colon X \to \uly$.
\item Each $g \colon X \to \uly$ can be composed with a $\cat{C}$-map $f \colon X' \to X$ or $\cat{D}$-map $h \colon \uly \to \uly'$.
\end{itemize}
For  $g \colon X \to \uly$ we must have the following.  (We use semicolon for diagrammatic-order composition.)
\begin{eqnarray*}
  \id_{X};g & = & g \\
(f';f);g & = & f';(f;g)\\
g;\id_{Y} & = & g \\
g;(h;h') & = & (g;h);h' \\
(f;g);h & = & f; (g;h)
\end{eqnarray*}
\end{definition}
For example, for an endofunctor  $H$ on $\catc$, the coalgebra-to-algebra morphisms constitute a bimodule $\CoAlg{H} \relto \Alg(H)$.  Bimodules $\catc \relto \catd$ corresponds to functors $\op{\catc} \times \cat{D} \to \Sets$ and are also called \emph{distributors} or \emph{profunctors} (but some authors reverse the direction).  
\begin{definition} \label{def:mapbim} \hfill
  \begin{enumerate}
  \item A \emph{map} of bimodules
    $$
      \xymatrix{
 \cat{C} \ar@/^10pt/[r]|-@{|}^{\cato} \ar@{}[r]|{R \Downarrow} \ar@/_10pt/[r]|-@{|}_{\cato'} & \cat{D}
} 
    $$
sends each $\cato$-morphism $g \colon X \to \uly$ to an $\cato'$-morphism $R g \colon X \to \uly$ with the following commuting:
\begin{displaymath}
  \xymatrix{
  X'\ar[d]_{f} \ar[dr]^{R(f;g)} & & X\ar[r]^{Rg} \ar[dr]_{R(g;h)} & \uly \ar[d]^{h} \\
 X \ar[r]_{Rg} & \uly & & \uly'
}
\end{displaymath}
\item \label{item:tcell} More generally, a \emph{2-cell} 
  \begin{displaymath} %\label{eqn:twocell}
    \xymatrix{
    \cat{C} \ar[d]_{H} \ar@{}[dr]|{R\,\Downarrow} \ar[r]|-@{|}^{\cato} & \cat{D} \ar[d]^{L} \\
    \cat{C'} \ar[r]|-@{|}_{\cato'} & \cat{D'}
    }
  \end{displaymath}
sends each $\cato$-morphism $g \colon X \to \uly$ to an $\cato'$-morphism $R g \colon HX \to L\uly$ with the following commuting:
\begin{displaymath}
  \xymatrix{
  HX'\ar[d]_{Hf} \ar[dr]^{R(f;g)} &  & HX\ar[r]^{Rg} \ar[dr]_{R(g;h)} & L\uly \ar[d]^{Lh} \\
 HX \ar[r]_{Rg} & L\uly & &L\uly'
}
\end{displaymath}
  \end{enumerate}
\end{definition}
The product construction on categories extends to bimodules:
\begin{definition} \label{def:prodbim}
  Let $(\cat{O}_j \colon \cat{C}_j \relto \cat{D}_j)_{j \in J}$ be a family of bimodules.  Then the bimodule $\prod_{j \in J} \cat{O}_j \colon \prod_{j \in J} \cat{C}_j \relto \prod_{j \in J} \cat{D}_j$ is defined by saying that a morphism $(X_j)_{j \in J} \to (\uly_{j})_{j \in J}$ is a family $(g_j \colon X_j \to \uly_j)_{j \in J}$, where $g_j$ is an $\cat{O}_j$-morphism for all $j \in J$
  .  Composition is defined componentwise.
\end{definition}
Of course this construction extends also to maps and 2-cells between bimodules.

Here are two ways of constructing a bimodule $\cat{C} \relto \cat{D}$.
\begin{definition} \label{def:lbimrbim} \hfill
\begin{enumerate}
\item \label{item:lbim} A functor $F \colon \cat{C} \to \cat{D}$ gives $\lbim{F} \colon \cat{C} \relto \cat{D}$, where $\lbim{F}(X,\uly) \eqdef \catd(FX,\uly)$.
\item \label{item:rbim} A functor $G \colon \cat{D} \to \cat{C}$ gives $\rbim{G} \colon \cat{C} \relto \cat{D}$, where $\rbim{G}(X,\uly) \eqdef \catc(X,G\uly)$.
\end{enumerate}
\end{definition}

\begin{definition} \label{def:bimrep}
For a bimodule $\cato \colon \catc \relto \catd$,
\begin{itemize}
\item a \emph{left representation} consists of a functor $F \colon \catc \to \catd$ and an isomorphism $m: \cato \cong \lbim{F}$
\item a \emph{right representation} consists of a functor $G \colon \catc \to \catd$ and an isomorphism $n: \cato \cong \rbim{G}$.
\end{itemize}
\end{definition}
Note that an adjunction
\begin{displaymath}
   \xymatrix{
\cat{C}\ar@/^2ex/[rr]^-{F} \save !L(.5)  \restore & \bot & \cat{D} \ar@/^2ex/[ll]^-G \save !R(.5)  \restore 
}
\end{displaymath}
may be viewed as a bimodule isomorphism $\lbim{F} \cong \rbim{G}$.  Conversely a bimodule $\catc \relto \catd$ equipped with both a left and a right representation constitutes an adjunction.

The natural transformations in Theorem~\ref{thm:matesgeneral} correspond to 2-cells of bimodules, as follows.
\begin{theorem} \label{thm:matesbim}
% \begin{enumerate}
% \item 
Suppose we have left representations $m\colon\cato \cong \lbim{F}$ and
$m'\colon\cato' \cong \lbim{F'}$.  Then a 2-cell
\begin{displaymath} %\label{eqn:twocell}
    \xymatrix{
    \cat{C} \ar[d]_{H} \ar@{}[dr]|{R\,\Downarrow} \ar[r]|-@{|}^{\cato} & \cat{D} \ar[d]^{L} \\
    \cat{C'} \ar[r]|-@{|}_{\cato'} & \cat{D'}
    }
  \end{displaymath}
corresponds to a natural transformation
\begin{displaymath}
  \xymatrix{
 \cat{C}\ar[r]^{F} \ar[d]_{H} \ar@{}[dr]|{\stackrel{\rho_2}{\Rightarrow}} &  \cat{D} \ar[d]^{L} \\
 \cat{C'} \ar[r]_{F'} & \cat{D'}
}
\end{displaymath}
where $R$ sends an $\cato$-morphism $g \colon X \to \uly$
 % corresponding to  $h \colon FX \to \uly$ to the $\cato'$-morphism $HX \to L\uly$ corresponding to the composite
to
  \begin{math}
  %{m'}^{-1} 
  \left(  \xymatrix{
  F'HX \ar[r]^-{\rho_2} & LFX \ar[r]^-{Lm(g)}  & L\uly
      }  \right)
  \end{math}.
The analogous statements hold for  $\rho_1$, $\rho_3$ and $\rho_4$.
\end{theorem}

Now we give a more refined account of steps.  Suppose we have a bimodule, two endofunctors and a 2-cell:
\begin{displaymath}
    \xymatrix{
    \cat{C} \ar[d]_{H} \ar@{}[dr]|{R\,\Downarrow} \ar[r]|-@{|}^{\cato} & \cat{D} \ar[d]^{L} \\
    \cat{C} \ar[r]|-@{|}_{\cato} & \cat{D}
    }
  \end{displaymath}
  We call $R$ a ``step''.  Given a left representation $m \colon \cato \cong \lbim
  {F}$ we have $\rho_2$ and the functor $F^{\rho}$.  Given a right representation $n \colon \cato \cong \rbim{G}$ we have $\rho_1$ and the functor $G_{\rho}$.

% Let us return to our basic setting (\ref{eq:adj-setting}), with the adjunction arising from a bimodule $\cato$ with a left and right representation.  By Theorem~\ref{thm:matesbim}, $\rho$ corresponds to a 2-cell
%  \begin{displaymath}
%     \xymatrix{
%     \cat{C} \ar[d]_{H} \ar@{}[dr]|{R\,\Downarrow} \ar[r]|-@{|}^{\cato} & \cat{D} \ar[d]^{L} \\
%     \cat{C} \ar[r]|-@{|}_{\cato} & \cat{D}
%     }
%   \end{displaymath}

\begin{definition}\hfill
  \begin{enumerate}
  \item A \emph{coalgebra morphism} from an $H$-coalgebra $c\colon
  X \rightarrow H(X)$ to an $L$-coalgebra $d \colon \Theta \rightarrow L(\Theta)$ is an $\cato$-morphism  $g \colon X \to \Theta$ such that the following commutes:
  \begin{displaymath}
    \xymatrix@R-0.5pc{
 X \ar[r]^-{f} \ar[d]_{c} & \Theta \ar[d]^{d}  \\
 H(X) \ar[r]_-{R(f)} & L(\Theta) 
}
  \end{displaymath}
This gives a bimodule $\CoAlg{H} \relto \CoAlg{L}$.
\item A \emph{coalgebra-to-algebra morphism} from an $H$-coalgebra $c\colon
  X \rightarrow H(X)$ to an $L$-algebra $a \colon L(\Theta) \Rightarrow \Theta$  is an $\cato$-morphism  $g \colon X \to \Theta$ such that the following commutes:
  \begin{displaymath}
    \xymatrix@R-0.5pc{
 X \ar[r]^-{f} \ar[d]_{c} & \Theta \\
 H(X) \ar[r]_-{R(f)} & L(\Theta) \ar[u]_{a}
}
  \end{displaymath}
Equivalently: such a morphism is a fixpoint for the
endofunction on the homset $\catc(X,\Theta)$ sending $f$ to the composite
\begin{math}
 \xymatrix{
 X \ar[r]^-{c} & H(X) \ar[r]^-{R(f)} & L(\Theta) \ar[r]^-{a} & \Theta
} 
\end{math}.
This gives a bimodule $\CoAlg{H} \relto \Alg{L}$.
  \end{enumerate}
\end{definition}

\begin{definition} \hfill
  \begin{enumerate}
  \item A final coalgebra $d \colon \Theta \Rightarrow L(\Theta)$ is said to \emph{extend across $\cato$} when from each $H$-coalgebra  $c \colon X \to H(X)$ there  is a unique morphism to $(\Theta,d)$.
\item A corecursive algebra  $a \colon L(\Theta) \Rightarrow \Theta$ is said to \emph{extend across $\cato$} when from each $H$-coalgebra  $c \colon X \to H(X)$ there  is a unique morphism to $(\Theta,a)$.
  \end{enumerate}
\end{definition}

Now let us decompose Proposition~\ref{prop:pres-corec} into two parts.
\begin{proposition} \label{prop:corec-bimod} \hfill
  \begin{enumerate}
  \item \label{item:corec-bimod-left} Let $\cato$ have a left representation $m\colon\cato \cong \lbim{F}$.  Then any corecursive $L$-algebra $(\Theta,a)$ %$a \colon L(\Theta) \Rightarrow \Theta$
extends across $\cato$. (And hence also any final $L$-algebra.)  Explicitly, the map $(X,c) \to (\Theta,a)$ is $m^{-1}$ applied to the map $F^{\rho}(X,c) \to (\Theta,a)$.
\item \label{item:corec-bimod-right} Let $\cato$ have a right representation $n \colon \cato \cong \rbim{G}$.  Then any corecursive $L$-algebra $(\Theta,a)$ extending across $\cato$ is sent by $G_{\rho}$ to a corecursive $H$-algebra.  Explicitly, the map $(X,c) \to G_{\rho} (\Theta,a)$ is $n$ applied to the map $(X,c) \to (\Theta,a)$.
  \end{enumerate}
\end{proposition}

% \begin{theorem}\hfill \label{thm:extendbim}
%   \begin{enumerate}
%   \item \label{item:extendcorecur} Any corecursive algebra $a \colon L(\Theta) \Rightarrow \Theta$ extends across $\cato$, i.e.\ from each $H$-coalgebra $c \colon X \to H(X)$ there is a unique morphism to $(\Theta,a)$.
%   \item \label{item:extendfinalbim} Any final coalgebra  $d \colon \Theta \Rightarrow L(\Theta)$ extends across $\cato$, i.e.\ from each $H$-coalgebra  $c \colon X \to H(X)$ there  is a unique morphism to $(\Theta,d)$.
%   \end{enumerate}
% \end{theorem}
% Part~(\ref{item:extendcorecur}) follows from left representability of $\cato$
% and part~(\ref{item:extendfinalbim}) follows.  Indeed part~(\ref{item:extendcorecur} can be seen as one half of Proposition~(\ref{prop:pres-corec}).  The other half says that any corecursive $L$-algebra extending across $\cato$ is sent by $G_{\rho}$ to a corecursive $H$-algebra; this follows from right representability of $\cato$.

Note that this story also appears, in contravariant form, in~\cite[Propositions 16--17]{Levy15}.

\end{appendix}

\end{document}

%% file: prooftree.tex
\newdimen\proofrulebreadth \proofrulebreadth=.05em
\newdimen\proofdotseparation \proofdotseparation=1.25ex
\newdimen\proofrulebaseline \proofrulebaseline=2ex
\newcount\proofdotnumber \proofdotnumber=3
\let\then\relax
\def\hfi{\hskip0pt plus.0001fil}
\mathchardef\squigto="3A3B
\newif\ifinsideprooftree\insideprooftreefalse
\newif\ifonleftofproofrule\onleftofproofrulefalse
\newif\ifproofdots\proofdotsfalse
\newif\ifdoubleproof\doubleprooffalse
\let\wereinproofbit\relax
\newdimen\shortenproofleft
\newdimen\shortenproofright
\newdimen\proofbelowshift
\newbox\proofabove
\newbox\proofbelow
\newbox\proofrulename
\def\shiftproofbelow{\let\next\relax\afterassignment\setshiftproofbelow\dimen0 }
\def\shiftproofbelowneg{\def\next{\multiply\dimen0 by-1 }%
\afterassignment\setshiftproofbelow\dimen0 }
\def\setshiftproofbelow{\next\proofbelowshift=\dimen0 }
\def\setproofrulebreadth{\proofrulebreadth}

\def\prooftree{% NESTED ZERO (\ifonleftofproofrule)
\ifnum  \lastpenalty=1
\then   \unpenalty
\else   \onleftofproofrulefalse
\fi
\ifonleftofproofrule
\else   \ifinsideprooftree
        \then   \hskip.5em plus1fil
        \fi
\fi
\bgroup% NESTED ONE (\proofbelow, \proofrulename, \proofabove,
\setbox\proofbelow=\hbox{}\setbox\proofrulename=\hbox{}%
\let\justifies\proofover\let\leadsto\proofoverdots\let\Justifies\proofoverdbl
\let\using\proofusing\let\[\prooftree
\ifinsideprooftree\let\]\endprooftree\fi
\proofdotsfalse\doubleprooffalse
\let\thickness\setproofrulebreadth
\let\shiftright\shiftproofbelow \let\shift\shiftproofbelow
\let\shiftleft\shiftproofbelowneg
\let\ifwasinsideprooftree\ifinsideprooftree
\insideprooftreetrue
\setbox\proofabove=\hbox\bgroup$\displaystyle % NESTED TWO
\let\wereinproofbit\prooftree
\shortenproofleft=0pt \shortenproofright=0pt \proofbelowshift=0pt
\onleftofproofruletrue\penalty1
}

\def\eproofbit{% NESTED TWO
\ifx    \wereinproofbit\prooftree
\then   \ifcase \lastpenalty
        \then   \shortenproofright=0pt  % 0: some other object, no indentation
        \or     \unpenalty\hfil         % 1: empty hypotheses, just glue
        \or     \unpenalty\unskip       % 2: just had a tree, remove glue
        \else   \shortenproofright=0pt  % eh?
        \fi
\fi
\global\dimen0=\shortenproofleft
\global\dimen1=\shortenproofright
\global\dimen2=\proofrulebreadth
\global\dimen3=\proofbelowshift
\global\dimen4=\proofdotseparation
\global\count255=\proofdotnumber
$\egroup  % NESTED ONE
\shortenproofleft=\dimen0
\shortenproofright=\dimen1
\proofrulebreadth=\dimen2
\proofbelowshift=\dimen3
\proofdotseparation=\dimen4
\proofdotnumber=\count255
}

\def\proofover{% NESTED TWO
\eproofbit % NESTED ONE
\setbox\proofbelow=\hbox\bgroup % NESTED TWO
\let\wereinproofbit\proofover
$\displaystyle
}%
\def\proofoverdbl{% NESTED TWO
\eproofbit % NESTED ONE
\doubleprooftrue
\setbox\proofbelow=\hbox\bgroup % NESTED TWO
\let\wereinproofbit\proofoverdbl
$\displaystyle
}%
\def\proofoverdots{% NESTED TWO
\eproofbit % NESTED ONE
\proofdotstrue
\setbox\proofbelow=\hbox\bgroup % NESTED TWO
\let\wereinproofbit\proofoverdots
$\displaystyle
}%
\def\proofusing{% NESTED TWO
\eproofbit % NESTED ONE
\setbox\proofrulename=\hbox\bgroup % NESTED TWO
\let\wereinproofbit\proofusing
\kern0.3em$
}

\def\endprooftree{% NESTED TWO
\eproofbit % NESTED ONE
  \dimen5 =0pt% spread of hypotheses
\dimen0=\wd\proofabove \advance\dimen0-\shortenproofleft
\advance\dimen0-\shortenproofright
\dimen1=.5\dimen0 \advance\dimen1-.5\wd\proofbelow
\dimen4=\dimen1
\advance\dimen1\proofbelowshift \advance\dimen4-\proofbelowshift
\ifdim  \dimen1<0pt
\then   \advance\shortenproofleft\dimen1
        \advance\dimen0-\dimen1
        \dimen1=0pt
        \ifdim  \shortenproofleft<0pt
        \then   \setbox\proofabove=\hbox{%
                        \kern-\shortenproofleft\unhbox\proofabove}%
                \shortenproofleft=0pt
        \fi
\fi
\ifdim  \dimen4<0pt
\then   \advance\shortenproofright\dimen4
        \advance\dimen0-\dimen4
        \dimen4=0pt
\fi
\ifdim  \shortenproofright<\wd\proofrulename
\then   \shortenproofright=\wd\proofrulename
\fi
\dimen2=\shortenproofleft \advance\dimen2 by\dimen1
\dimen3=\shortenproofright\advance\dimen3 by\dimen4
\ifproofdots
\then
        \dimen6=\shortenproofleft \advance\dimen6 .5\dimen0
        \setbox1=\vbox to\proofdotseparation{\vss\hbox{$\cdot$}\vss}%
        \setbox0=\hbox{%
                \advance\dimen6-.5\wd1
                \kern\dimen6
                $\vcenter to\proofdotnumber\proofdotseparation
                        {\leaders\box1\vfill}$%
                \unhbox\proofrulename}%
\else   \dimen6=\fontdimen22\the\textfont2 % height of maths axis
        \dimen7=\dimen6
        \advance\dimen6by.5\proofrulebreadth
        \advance\dimen7by-.5\proofrulebreadth
        \setbox0=\hbox{%
                \kern\shortenproofleft
                \ifdoubleproof
                \then   \hbox to\dimen0{%
                        $\mathsurround0pt\mathord=\mkern-6mu%
                        \cleaders\hbox{$\mkern-2mu=\mkern-2mu$}\hfill
                        \mkern-6mu\mathord=$}%
                \else   \vrule height\dimen6 depth-\dimen7 width\dimen0
                \fi
                \unhbox\proofrulename}%
        \ht0=\dimen6 \dp0=-\dimen7
\fi
\let\doll\relax
\ifwasinsideprooftree
\then   \let\VBOX\vbox
\else   \ifmmode\else$\let\doll=$\fi
        \let\VBOX\vcenter
\fi
\VBOX   {\baselineskip\proofrulebaseline \lineskip.2ex
        \expandafter\lineskiplimit\ifproofdots0ex\else-0.6ex\fi
        \hbox   spread\dimen5   {\hfi\unhbox\proofabove\hfi}%
        \hbox{\box0}%
        \hbox   {\kern\dimen2 \box\proofbelow}}\doll%
\global\dimen2=\dimen2
\global\dimen3=\dimen3
\egroup % NESTED ZERO
\ifonleftofproofrule
\then   \shortenproofleft=\dimen2
\fi
\shortenproofright=\dimen3
\onleftofproofrulefalse
\ifinsideprooftree
\then   \hskip.5em plus 1fil \penalty2
\fi
}